\documentclass[final,12pt]{colt2024} % Include author names

\title[Superconstant Inapproximability of Decision Tree Learning]{Superconstant Inapproximability of Decision Tree Learning}
\usepackage{times}

\coltauthor{\Name{Caleb Koch} \Email{ckoch@stanford.edu}\\
 \Name{Carmen Strassle} \Email{strassle@stanford.edu}\\
 \Name{Li-Yang Tan} \Email{liyang@cs.stanford.edu}\\
 \addr Stanford University}

\usepackage{tcolorbox}
\usepackage{liyang}
\usepackage{enumitem}
\usepackage{caption}

\usepackage{tikz}
\usetikzlibrary{calc,through,backgrounds,decorations.pathreplacing, calligraphy,arrows.meta}
\usetikzlibrary{positioning,chains,fit,shapes}

\usepackage{tikz-cd}

\usepackage{pgfplots}
\pgfplotsset{width=8cm,compat=newest}

\def\colorful{0}

\ifnum\colorful=1
\newcommand{\violet}[1]{{\color{violet}{#1}}}

\fi
\ifnum\colorful=0
\newcommand{\violet}[1]{{{#1}}}

\fi

\newcommand{\error}{\mathrm{error}}

\newcommand{\Sens}{\mathrm{Sens}}

\newcommand{\IsEdge}{\mathrm{\sc IsEdge}}

\def\SAT{{\sc SAT}} 

\newcommand{\NP}{\mathrm{NP}} 
 
\newcommand{\VC}{\mathrm{VC}}

\newcommand{\Cert}{\mathrm{Cert}}

\newcommand{\dtsize}{\mathrm{DT}} 
\newcommand{\ind}{\mathrm{Ind}}

\begin{document}
%!TEX root = paper.tex

\title{Superconstant inapproximability of decision tree learning \vspace{10pt}}

% Need to redefine citep for the arxiv version:
\newcommand{\citep}{\cite}

\author{ 
Caleb Koch \vspace{6pt} \\ 
{{\sl Stanford}} \and 
\hspace{5pt} Carmen Strassle \vspace{6pt} \\
\hspace{5pt} {{\sl Stanford}} \vspace{10pt} \and 
Li-Yang Tan \vspace{6pt}  \\
\hspace{-10pt} {{\sl Stanford}}
}

\date{\small{\today}}

\maketitle

\begin{abstract}
    We consider the task of properly PAC learning decision trees with queries. Recent work of Koch, Strassle, and Tan showed that the strictest version of this task, where the hypothesis tree~$T$ is required to be {\sl optimally} small, is NP-hard. Their work leaves open the question of whether the task remains intractable if $T$ is only required to be close to optimal, say within a factor of $2$, rather than exactly optimal. 

We answer this affirmatively and show that the task indeed remains NP-hard even if~$T$ is allowed to be within {\sl any} constant factor of optimal. More generally, our result allows for a smooth tradeoff between the hardness assumption and inapproximability factor. As Koch et al.’s techniques do not appear to be amenable to such a strengthening, we first recover their result with a new and simpler proof, which we couple with a new XOR lemma for decision trees. While there is a large body of work on XOR lemmas for decision trees, our setting necessitates parameters that are extremely sharp and are not known to be attainable by existing such lemmas. Our work also carries new implications for the related problem of {\sc Decision Tree Minimization}. 
\end{abstract}

% \begin{keywords}%
%    Decision trees, hardness of approximation, learning with queries
% \end{keywords}

 \thispagestyle{empty}
 \newpage 
 \setcounter{page}{1}

\section{Introduction}

Decision trees are a basic and popular way to represent data. Their simple logical structure makes them the prime example of an interpretable model. They are also the base model at the heart of powerful ensemble methods, such as XGBoost and random forests, that achieve state-of-the-art performance in numerous settings. Owing in part to their practical importance, the task of efficiently constructing decision tree representations of data has been intensively studied in the theory community for decades, in a variety of models and from both algorithmic and hardness perspectives. Indeed, on the heels of Cook and Karp’s papers on the theory of NP-hardness,~\cite{HR76} proved that a certain formulation of decision tree learning is NP-hard. Quoting their introduction, ``While the proof to be given is relatively simple, the importance of this result can be measured in terms of the large amount of effort that has been put into finding efficient algorithms for constructing optimal binary decision trees." This effort has only compounded over the years, with a recent surge of interest coming from the interpretable machine learning community; the 2022 survey~\citep{RCCHSZ22} lists decision tree learning as the very first of the field’s ``10 grand challenges".

We consider the problem within the model of PAC learning with queries~\citep{Val84,Ang88}. In this model, the learner is given query access to a function $f$ and i.i.d.~draws from a distribution $\mathcal{D}$, along with the promise that $f$ is computable by a size-$s$ decision tree. Its task is to output a size-$s'$ decision tree that achieves high accuracy with respect to $f$ under $\mathcal{D}$, where $s'$ is as close to~$s$ as possible. Motivation for studying query learners for this problem is twofold. First, it models the task of converting an {\sl existing} trained model $f$, for which one has query access to, into its decision tree representation---a common post-processing step for interpretability reasons. The second, more intrinsic, motivation comes from the fact that computational lower bounds against query learners, for {\sl any} learning task, have generally been  elusive. We are aware of only one such result outside of decision tree learning, on the NP-hardness of learning DNF formulas with queries~\citep{Fel06}---this resolved a longstanding open problem of~\cite{Val84,Val85}.

\paragraph{\cite{KST23}.} For decision tree learning, recent work of Koch, Strassle, and Tan showed that the strictest version of the problem, where $s' = s$, is NP-hard. This  resolved an open problem that had been raised repeatedly over the years~\citep{Bsh93,GLR99,MR02,Fel16}, but still left open the possibility of efficient algorithms achieving $s'$ that is slightly larger than $s$. \cite{KST23} listed this as an natural avenue for further research, while also pointing to challenges in extending their techniques to even rule out $s' = 2s$. 

\subsection{This work} We show that the problem remains NP-hard even for $s' = Cs$ where $C$ is an arbitrarily large constant:

\medskip 

\begin{tcolorbox}[colback = white,arc=1mm, boxrule=0.25mm]
\begin{theorem}
\label{thm:main}
For every constant $C>1$, there is a constant $\eps>0$ such that the following holds. If there is an algorithm running in time $t(n)$ that, given queries to an $n$-variable function $f$ computable by a decision tree of size $s = O(n)$ and random examples $(\bx,f(\bx))$ drawn according to a distribution $\mathcal{D}$, outputs w.h.p.~a decision tree of size $Cs$ that is $\eps$-close to $f$ under $\mathcal{D}$, then \textnormal{SAT} can be solved in randomized time $O(n^2)\cdot t(\poly(n))$. 
\end{theorem}
\end{tcolorbox}
\medskip 

Consequently, assuming $\mathrm{NP} \ne \mathrm{RP}$, any algorithm for the problem has to either be inefficient with respect to time (i.e.~take superpolynomial time), or inefficient with respect to representation size (i.e.~output a hypothesis of size much larger than actually necessary).  

\Cref{thm:main} is a special case of a more general result that allows for a smooth tradeoff between the strength of the hardness assumption on one hand and the inapproximability factor on the other hand:

\medskip 

\begin{tcolorbox}[colback = white,arc=1mm, boxrule=0.25mm]
\begin{theorem}
    \label{thm:main-general}
     Suppose for some $r\ge 1$ there is a time $t(s,1/\eps)$ algorithm which given queries to an $n$-variable function $f$ computable by a decision tree of size $s$ and random examples $(\bx,f(\bx))$ drawn according to a distribution $\mathcal{D}$, outputs w.h.p.~a decision tree of size $2^{O(r)}\cdot s$ that is $\eps$-close to $f$ under $\mathcal{D}$. Then $\mathrm{SAT}$ can be solved in randomized time $\Tilde{O}(rn^2)\cdot t(n^{O(r)},2^{O(r)})$. 
\end{theorem}
\end{tcolorbox}
\medskip 

By taking $r$ to be superconstant in \Cref{thm:main-general}, we obtain superconstant inapproximability ratios at the price of stronger yet still widely-accepted hardness assumptions. For example, assuming SAT cannot be solved in randomized {\sl quasipolynomial} time, we get a near-polynomial inapproximability ratio of $2^{(\log s)^\gamma}$ for any constant $\gamma < 1$. % decision tree hypotheses. This is obtained by choosing $r=\polylog n$ in \Cref{thm:main-general}.}
% \lnote{Add a takeaway setting of parameters: ``Consequently, assuming SAT not in RTIME[quasipoly], we have that [...]}

Our work also carries new implications for the related problem of {\sc Decision Tree Minimization}: Given a decision tree $T$, construct an equivalent decision tree $T'$ of minimal size. This problem was first shown to be NP-hard by~\cite{ZB00}, and subsequently~\cite{Sie08} showed that it is NP-hard even to approximate. We recover~\cite{Sie08}'s inapproximability result, and in fact strengthen it to hold even if $T'$ is only required to {\sl mostly} agree with $T$ on a given {\sl subset} of inputs (rather than fully agree with $T$ on all inputs as in~\cite{Sie08}).   See \Cref{sec:dt-dataset-min} for details. 

%\begin{theorem}
%\label{thm:dtmin-intro}
%\lnote{Currently thinking of removing this from the intro. Let's discuss.} For all constants $C>1$, \violet{there is a constant $\eps>0$} such that the following task is \textnormal{NP}-hard. Given a decision tree $T$ over $n$ variables, a \violet{set} $D\sse\zo^n$ of size $O(n)$, and a parameter $s=O(n)$, distinguish between
%    \begin{itemize}
%        \item[$\circ$] {Yes case}: there is a size-$s$ decision tree $T'$ such that $T'(x)=T(x)$ for all $x\in \zo^n$; and
%        \item[$\circ$] {No case}: all decision trees $T'$ such that $T'(x)=T(x)$ for a \violet{$(1-\eps)$-fraction of $x\in D$} have size at least $Cs$.
 %   \end{itemize}
%\end{theorem}

\section{Background and Context}

\subsection{Algorithms for properly learning decision trees}

In the language of learning theory, we are interested in the task of properly PAC learning decision trees. We distinguish between {\sl strictly-proper} learning, where the size $s'$ of the hypothesis decision tree has to exactly match the optimal size $s$ of the target decision tree, and {\sl weakly-proper} learning, where $s'$ can be larger than $s$.~\cite{KST23} therefore establishes the hardness of strictly-proper learning whereas our work establishes the hardness even of weakly-proper learning even when $s'$ is larger than $s$ by any constant.
%\carnote{should we clarify here that we establish the hardness of weakly-proper learning when $s'$ is larger than $s$ by a constant} 
We now overview the fastest known algorithms for both settings. \vspace{-7pt}

\paragraph{Strictly-proper learning via dynamic programming.} There is a simple $2^{O(n)}$ time algorithm for strictly-properly learning decision trees: draw a dataset of size $O(s\log n)$, dictated by the VC dimension of size-$s$ decision trees, and run a $2^{O(n)}$-time dynamic program to find a size-$s$ decision tree that fits the dataset perfectly. (See~\cite{GLR99,MR02} for a description of this dynamic program.) There are no known improvements to this naive algorithm if one insists on strictly-proper learning, and indeed,~\cite{KST23}'s result strongly suggests that there are probably none. %\cite{KST23}’s NP-hardness result shows that this naive algorithm is essentially optimal if one insists on strictly-proper learners: assuming SAT requires exponential time, so does strictly-properly learning decision trees.
\vspace{-7pt}

\paragraph{Weakly-proper learning via Ehrenfeucht--Haussler.} The setting of weakly-proper learning allows for a markedly faster algorithm: a classic algorithm of~\cite{EH89} runs in $n^{O(\log s)}$ time and outputs a decision tree hypothesis of size $s' = n^{O(\log s)}$. \cite{EH89} listed as an open problem that of designing algorithms that output smaller hypotheses, i.e.~ones where $s'$ is closer to $s$. There has been no algorithmic progress on this problem since 1989, and prior to our work, there were also no hardness results ruling out efficient algorithms achieving say $s' = 2s$.

\subsection{Lower bounds for random example learners}

\begin{figure}[!ht]

\begin{tcolorbox}[colback = white,arc=1mm, boxrule=0.25mm]

    \centering
    \begin{tikzpicture}[]
        % x DEFINES HOW WIDE FIGURE IS
        \def\x{5}
        % c DEFINES CENTER
        \def\c{0}
        % t DEFINES TOP
        \def\t{6}
        % t DEFINES TOP OF VERTICAL AXIS
        \def\tt{5}
        % MAIN VERTICAL AXIS
        \draw[black, |->] (\c,-4) node[right] {} -- (\c,\tt) node[right]{};

        % LABEL INAPPROXIMABILITY FACTOR TOP
        \draw[xshift=0 cm] (\c,\t) node[above,fill=white,text width=4cm,text centered] { {\textbf{Inapproximability factor}}};

        % LABELS for RANDOM SAMPLES/QUERIES
        \draw[xshift=-\x cm] (\c,\t) node[above,fill=white,text width=6cm,text centered] { {\textbf{Lower bounds against random-example learners}}};
        \draw[xshift=\x cm] (\c,\t) node[above,fill=white,text width=5cm,text centered] { {\textbf{Lower bounds against query learners}}};

        % UNDERLINES
        % \draw[xshift=-7.75 cm,black, -] (\c,\t) node[below] {} -- (5.5,\t) node[below]{};
        % \draw[xshift=7.75 cm,black, -] (\c,\t) node[below] {} -- (-5.5,\t) node[below]{};

        % LINES
        \draw[black, -] (\c,-3.5) node[] {} -- (-4,-3.5) node[]{};
        \draw[black, -] (\c,-1.5) node[] {} -- (-4,-1.5) node[]{};
        \draw[black, -] (\c,.5) node[] {} -- (-4,.5) node[]{};
        \draw[black, -] (\c,2.5) node[] {} -- (-4,2.5) node[]{};
        \draw[black, -] (\c,4.5) node[] {} -- (-4,4.5) node[]{};

        \draw[black, -] (\c,-3.5) node[] {} -- (4,-3.5) node[]{};
        \draw[black, -] (\c,-1.5) node[] {} -- (4,-1.5) node[]{};

        % LEVELS OF INAPPROXIMABILITY
        \draw[color=black] (\c,-3.5) node[fill=white,text width=3.5cm,text centered] {\color{black} Exact \\ ({\footnotesize i.e.~strictly-proper learning})};
        \draw[color=black] (\c,-1.5) node[fill=white,text width=3.5cm,text centered] {\color{black} Superconstant};
        \draw[color=black] (\c,.5) node[fill=white,text width=3.5cm,text centered] {\color{black} Polynomial};
        \draw[color=black] (\c,2.5) node[fill=white,text width=3.5cm,text centered] {\color{black} Superpolynomial};
        \draw[color=black] (\c,4.5) node[fill=white,text width=3.5cm,text centered] {\color{black} Quasipolynomial \\ ({\footnotesize matching~\cite{EH89}'s guarantee})};
        % RANDOM SAMPLE REFS
        \draw[color=black,xshift=-\x cm] (\c,-3.5) node[fill=white,text width=4.1cm,text centered] {\color{black} \cite{PV88} \\ (under $\NP\neq \mathrm{RP}$)};
        \draw[color=black,xshift=-\x cm] (\c,-1.5) node[fill=white,,text width=4.1cm,text centered] {\color{black} \cite{HJLT96} \\ (under $\NP\neq\mathrm{RP}$)};
        \draw[color=black,xshift=-\x cm] (\c,.5) node[fill=white,text width=4.1cm,text centered] {\color{black} \cite{ABFKP09}\\ (under ETH)};
        \draw[color=black,xshift=-\x cm] (\c,2.5) node[fill=white,text width=4.1cm,text centered] {\color{black} \cite{KST23soda}\\ (under ETH)};
        \draw[color=black,xshift=-\x cm] (\c,4.5) node[fill=white,text width=4.1cm,text centered] {\color{black} \cite{KST23soda}\\ (under inapproximability of parameterized {\sc Set Cover})};

        % QUERY REFS
        \draw[color=black,xshift=\x cm] (\c,-3.5) node[fill=white,text width=4.1cm,text centered] {\color{black} \cite{KST23} \\ (under $\NP\neq \mathrm{RP}$)};
        \draw[color=black,xshift=\x cm] (\c,-1.5) node[fill=white,text width=4.1cm,text centered] {\color{teal} \textbf{This work  \\ (under $\NP\ne \mathrm{RP}$) }};
        % \draw[color=black,xshift=\x cm] (\c,.5) node[fill=white] {\color{red} \Large \textbf{ }};
        % \draw[color=black,xshift=\x cm] (\c,2.5) node[fill=white] {\color{red} \Large \textbf{ }};
        % \draw[color=black,xshift=\x cm] (\c,4.5) node[fill=white] {\color{red} \Large \textbf{ }};

        % INAPPROXIMABILITY LABEL
        % \draw[color=black] (\c,-4.5) node[fill=white] {\color{black} \textbf{Inapproximability factor}};

    \end{tikzpicture}

\end{tcolorbox}
    
    \caption{Summary of lower bounds for decision tree learning.}
    \label{fig:lowerbounds}
\end{figure}
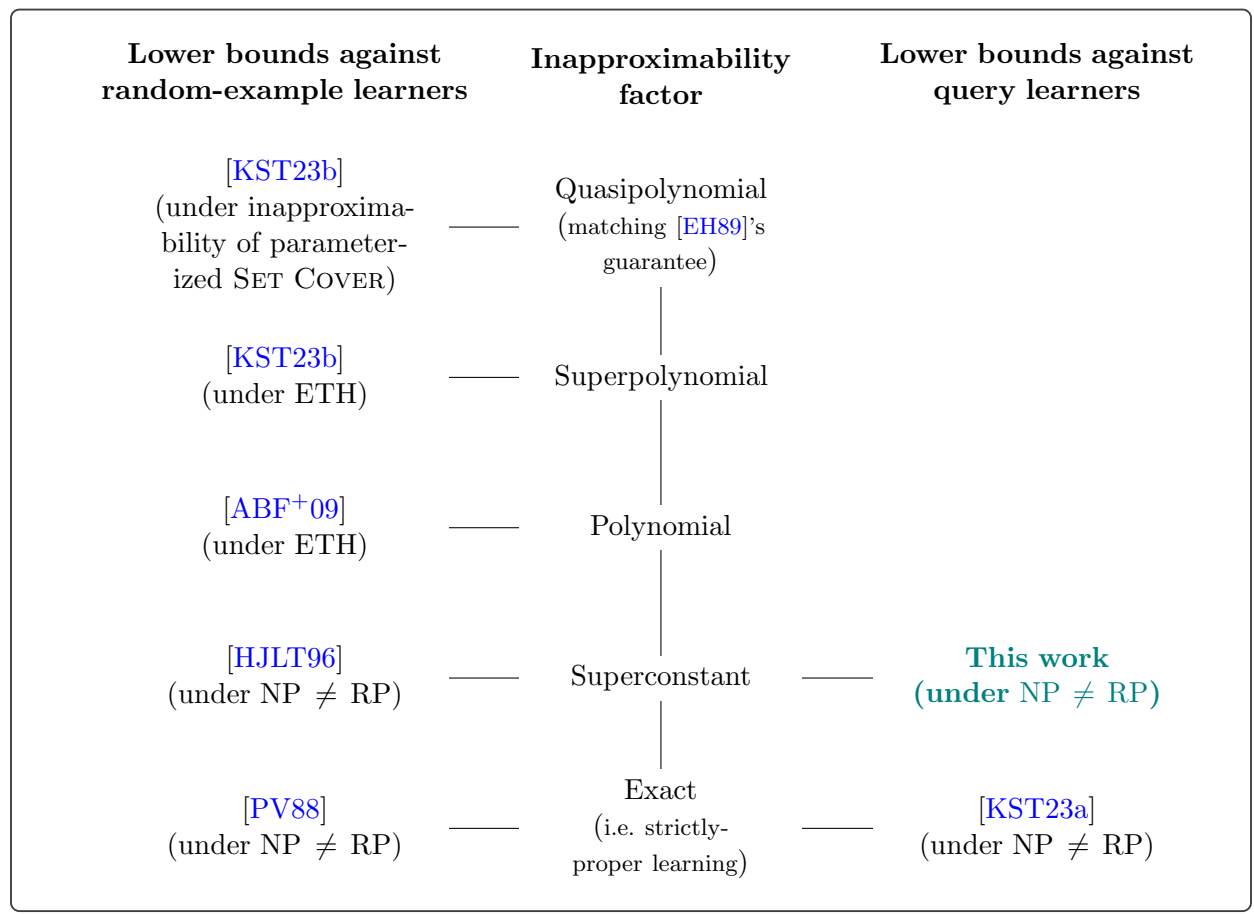

The problem is also well-studied in the model of PAC learning from {\sl random examples}, where the algorithm is only given labeled examples $(\boldsymbol{x},f(\boldsymbol{x}))$ where $\boldsymbol{x} \sim \mathcal{D}$. Lower bounds against random example learners are substantially easier to establish, and a sequence of works has given strong evidence of the optimality~\cite{EH89}’s weakly-proper algorithm under standard complexity-theoretic assumptions.

\cite{PV88}, in an early paper on the hardness of PAC learning, showed that strictly-proper learning of decision trees from random examples is NP-hard; they attributed this result to an unpublished manuscript of~\cite{Ang}.~\cite{HJLT96} then established a superconstant inapproximability factor (assuming $\mathrm{NP} \ne \mathrm{RP}$), which was subsequently improved to polynomial by~\cite{ABFKP09} (assuming the Exponential Time Hypothesis (ETH)). Recent work of~\cite{KST23soda} further improves the inapproximability factor to superpolynomial (assuming ETH) and quasipolynomial (assuming the inapproximability of parameterized {\sc Set Cover}), the latter of which exactly matches~\cite{EH89}’s performance guarantee.

It is reasonable to conjecture that~\cite{EH89}'s algorithm is optimal even for query learners. If so, our work is a step forward for proving lower bounds in the more challenging setting of query learners so that these bounds might ``catch up" with those in the random example setting; historically, the race has not been close---query-learner lower bounds have lagged far behind.
%is a step towards getting lower bounds in the more challenging setting of query learners “caught up” with those in the random example setting, for which they have thus far lagged far behind. 
Just as~\cite{KST23} can be viewed as establishing the query-learner analogue of~\cite{PV88}’s result (i.e.~the hardness of strictly-proper learning), our work can be viewed as establishing 
 the query-learner analogue of~\cite{HJLT96}’s result (i.e.~superconstant inapproximability of weakly-proper learning).  \Cref{fig:lowerbounds} summarizes the current landscape of decision tree lower bounds and shows how our work fits into it.

We refer the reader to Section 2.1 of~\cite{KST23} for a discussion of the technical challenges involved in proving lower bounds against query learners, and why the techniques developed in~\citep{PV88,HJLT96,ABFKP09,KST23soda} for random-example learners do not carry over.

\subsection{Other related work: improper learning of decision trees}

There is also vast literature on {\sl improper} learning of decision trees, where the target function is assumed to be a small decision tree but the hypothesis does not have to be one (see e.g.~\citep{Riv87,Blu92,KM93,Han93,Bsh93,BFJKMR94,HJLT96,JS06,OS07,KS06,GKK08,KST09,HKY18,CM19}). Examples of hypotheses that are constructed by existing algorithms include the sign of low-degree polynomials and small-depth boolean circuits. 

%We remark that in the machine learning literature, “decision tree learning” almost exclusively refers to the problem of constructing decision tree {\sl hypotheses}. See e.g.~the~\cite{Wiki-DTlearn} or Chapter 18 of the textbook~\cite{SS14}.  

\section{Technical Overview}

At a high level, our proof proceeds in two steps:

\begin{itemize}[leftmargin=12pt]
\item[$\circ$] {\bf Step 1: Slight inapproximaibility.} We first give a new proof of~\cite{KST23}’s result. In fact, we prove a statement that is (very)  slightly stronger than  the hardness of strictly-proper learning: we show that it is NP-hard for query learners to construct a decision tree of size $s' = (1+\delta)\cdot s$ for small constant $\delta < 1$.
%\carnote{should we mention here that $c<1$ so we still don't get $s$ versus $2s$}  
 While such a slight strengthening is not of much independent interest, it is important for technical reasons because it establishes {\sl some} inapproximability factor, albeit a small one, which we then amplify in the next step. 

\item[$\circ$] {\bf Step 2: Gap amplification.} We give a reduction that for any integer $r$ runs in time $n^{O(r)}$ and amplifies the inapproximability factor of $s'/s = 1+\delta$ from the step above into $(1+\delta)^r$. In particular, for any arbitrarily large constant $C$ this is a reduction that runs in polynomial time and amplifies the inapproximability factor to $C$. 

At the heart of this reduction is a new XOR lemma for decision trees: roughly speaking, this lemma says that if decision trees of size $s'$ incur large error when computing $f$, then decision trees of size $(s')^r$ incur large error when computing the $r$-fold XOR $f^{\oplus r}(x^{(1)},\ldots,x^{(r)}) \coloneqq f(x^{(1)})\oplus \cdots \oplus f(x^{(r)})$. 
%this lemma says that if $f$ cannot be approximated by decision trees of size $s'$, then its $r$-fold XOR $f^{\oplus r}(x^{(1)},\ldots,x^{(r)}) \coloneqq f(x^{(1)})\oplus \cdots \oplus f(x^{(r)})$ cannot be approximated by decision trees of size $(s')^r$. 
\end{itemize}

There is a large body of work on XOR lemmas for decision tree complexity~\citep{IRW94,NRS94,Sav02,Sha04,KSdW07,JKS10,Dru12,BDK18,BB19,BKLS20}, but our setting necessitates extremely sharp parameters that are not known to be achievable by any existing ones. Most relevant to our setting is one by~\cite{Dru12}, but it only reasons about the error of decision trees of size $(s')^{cr}$ for some $c < 1$ instead of $(s')^r$. This constant factor loss is inherent to~\cite{Dru12}'s proof technique, and we explain in~\Cref{remark:sharp} why we cannot afford even a tiny constant factor loss in the exponent.% \medskip 

\renewcommand{\Ind}{\mathrm{Ind}}

\subsection{Step 1: Slight inapproximability} 

Like~\cite{KST23}, we reduce from the NP-complete problem {\sc Vertex Cover}. For every graph $G$ there is an associated {\sl edge indicator function} $\IsEdge$ defined as follows:

%We sketch the main ideas behind this core reduction. For an $n$-vertex graph $G$, we consider its  {\sl edge indicator function} $\mathrm{\IsEdge}_G : \zo^n \to \zo$. An input $v = (v_1,\ldots,v_n) \in \zo^n$ to $\mathrm{\IsEdge}_G$ is viewed as specifying the presence or absence of each vertex $v_1,\ldots,v_n \in V$, and $\mathrm{\IsEdge}_G(v) = 1$ iff $v$ specifies the presence of exactly the two endpoints of some edge of $G$. More formally: 

\begin{definition}[$\mathrm{\IsEdge}_G$]\label{def: isedge}
Let $G$ be a graph with vertex set $V=\{v_1,\ldots,v_n\}$. We write $\Ind[e] \in \zo^n$ for the encoding of an edge $e\in E$ in $\zo^n$. That is, $\Ind[e]_i=1$ if and only if the vertex $v_i$ is in $e$. The {\sl edge indicator function 
of $G$} is the function $\mathrm{\IsEdge}_G : \zo^n \to \zo$, 
\[ \mathrm{\IsEdge}_G(x) =
\begin{cases} 
1 & x =\Ind[e] \text{ for some $e \in E$} \\
0 & \text{otherwise.} 
\end{cases}
\]
\end{definition}

For technical reasons, we work with a generalization of $\IsEdge$ called $\ell$-$\IsEdge$ where $\ell \in \N$ is a tuneable ``padding parameter" (we defer the definition of $\ell$-$\IsEdge$ to the main body; see \Cref{def:ell-isedge}). We prove that the decision tree complexity of $\ell\text{-}\mathrm{\IsEdge}_G$ scales with the vertex cover complexity of $G$ with fairly tight quantitative parameters: 

\begin{claim}
\label{claim:slight-inapprox}
    Let $G$ be a graph on $n$ vertices and $m$ edges. For all $\ell\ge 1$ and $\eps>0$, the following two cases hold.
    \begin{itemize}
        \item[$\circ$] Yes case: if $G$ has a vertex cover of size $k$, then there is a decision tree $T$ computing $\ell\text{-}\mathrm{\IsEdge}_G$ whose size satisfies
        $$
        |T|\le (\ell+1)(k+m)+mn.
        $$
        \item[$\circ$] No case: there is a distribution $\mathcal{D}$ such that if every vertex cover of $G$ has size at least $k'$, then any decision tree $T$ that is $\eps$-close to $\ell\text{-}\mathrm{\IsEdge}_G$ over $\mathcal{D}$ has size at least
        $$
        |T|\ge (\ell+1)\left(k'+(1-4\eps)m\right).
        $$
    \end{itemize}
\end{claim}

It is known that there is a constant $\delta >0$ such that deciding whether a graph has a vertex cover of size $\le k$ or requires vertex cover size $\ge (1+\delta)k$ is NP-hard \citep{PY91,Has01,DS05}. With an appropriate choice of parameters,~\Cref{claim:slight-inapprox} translates this into a gap of $\le s$ versus $\ge (1+\delta')s$ for some other constant $\delta'>0$ in the decision tree complexity of $\ell\text{-}\mathrm{\IsEdge}$. The NP-query hardness of learning size-$s$ decision trees with hypotheses of size $(1+\delta')s$ follows as a corollary.

% {\sc Vertex Cover} is NP-hard to approximate to within a factor of $1.01$. Using this hardness and \Cref{claim:slight-inapprox}, we can obtain NP-hardness of learning size-$s$ decision trees with size-$1.001s$ decision tree hypotheses. 

% \lnote{Say something about how Vertex Cover is known to be hard to approximate to within $1.01$, and the claim above allows us to get $s$ vs. $1.0001$ hardness.}

\paragraph{Key ingredients in the proof of~\Cref{claim:slight-inapprox}: Patch up and hard distribution lemmas.} As is often the case in reductions such as \Cref{claim:slight-inapprox}, the upper bound in the Yes case is straightforward to establish and most of the work goes into proving the lower bound in the No case. \cite{KST23}’s analysis of their No case is rather specific to the $\ell\text{-}\IsEdge$ function, whereas we develop a new technique for proving such lower bounds on decision tree complexity. In addition to being more general and potentially useful in other settings, our technique lends itself to an ``XOR-ed generalization" which we will need for gap amplification. (\cite{KST23}'s technique does not appear to be amenable to such a generalization, despite our best efforts at obtaining it.\footnote{On a more technical level,~\cite{KST23}'s technique requires them to reason about the complexity of {\sc Partial Vertex Cover}, a generalization of {\sc Vertex Cover}, whereas our simpler approach bypasses the need for this.})  

There are two components to our technique, both of which are generic statements concerning a decision tree $T$ that imperfectly computes a function $f$. The first is a {\sl patch up lemma} that shows how $T$ can be patched up so that it computes $f$ perfectly. The cost of this patch up operation, i.e.~how much larger $T$ becomes, is upper bounded by the certificate complexity of $f$, a basic and well-studied complexity measure of functions. (We defer the formal definitions of the technical terms and notation used in these lemmas to~\Cref{sec:prelims}.)     

\begin{lemma}[Patch up lemma]
% \label[lemma]{mylemma}
    \label[lemma]{lem:patchup-intro}
    Let $f:\zo^n\to\zo$ be a function and let $T$ be a decision tree. Then
    $$
    \dtsize(f)\le |T|+\sum_{x\in f^{-1}(1)}\Cert(f_{\pi(x)},x)
    $$
    where $\pi(x)$ denote the path followed by $x$ in $T$ and $f_{\pi(x)}$ is the restriction of $f$ by $\pi(x)$. 
\end{lemma}

The second component is a {\sl hard distribution lemma} that shows how a hard distribution $\mathcal{D}$ can be designed so that the error of $T$ with respect to $f$ under $\mathcal{D}$ is large. Roughly speaking, the more weight that $\mathcal{D}$ places on ``highly sensitive" points, the larger the error is:

\begin{lemma}[Hard distribution lemma]
\label[lemma]{lem:hard-distribution-intro}
Let $f:\zo^n\to\zo$ be a nonconstant function. Then for all nonempty $C \sse f^{-1}(1)$, there is a distribution over $C$ and all of its sensitive neighbors such that for any decision tree $T$, we have
$$
\error_{\mathcal{D}}(T,f)\ge \frac{1}{2|C|\Sens(f)}\sum_{x\in S}|\Sens(f_{\pi(x)},x)|
$$
where $\pi(x)$ is the path followed by $x$ in $T$ and $f_{\pi(x)}$ is the restriction of $f$ by $\pi(x)$. 
\end{lemma}

% (We defer the definitions of $\Cert(f,x)$ and $\Sens(f,x)$ to the body of the paper; see \Cref{para:boolean-functions}.) 

The No case of~\Cref{claim:slight-inapprox} follows by applying~\Cref{lem:patchup-intro,lem:hard-distribution-intro} to the $\ell$-$\IsEdge$ function and reasoning about its certificate complexity and sensitivity. 

\subsection{Step 2: Gap amplification} As alluded to above, a key advantage of our approach is that the patch up and hard distribution lemmas lend themselves to “XOR-ed generalizations”:

\begin{lemma}[XOR-ed version of Patch Up Lemma, see \Cref{lem:patchup-xor} for the exact version] 
\label[lemma]{lem:patchup-xor-intro}
    Let $f:\zo^n\to\zo$ be a function and let $T$ be a decision tree. Then for all $r\ge 1$,
    $$
    \dtsize(f^{\oplus r})\le |T|+2^r\sum_{x\in f^{-1}(1)^r}\prod_{i=1}^r\max\{1,\Cert(f_{\pi(x)},x^{(i)})\}
    $$
    where $\pi(x)$ is the path followed by $x$ in $T$ and $f_{\pi(x)}$ is the restriction of $f$ by $\pi(x)$.
\end{lemma}

\begin{lemma}[XOR-ed version of Hard Distribution Lemma, see \Cref{lem:hard-distribution-lemma-xor} for the exact version]
\label[lemma]{lem:hard-distribution-xor-intro}
    \violet{Let $f:\zo^n\to\zo$ be a nonconstant function, $C\sse f^{-1}(1)$ be nonempty, and $T$ be a decision tree.} There is a distribution $\mathcal{D}$ over the inputs in $C$ and their sensitive neighbors such that for all $r\ge 1$, 
    $$
    \error_{\mathcal{D}^{\otimes r}}(T,f^{\oplus r})\ge \left(\frac{1}{2|C|\Sens(f)}\right)^r\sum_{x\in C^r}\prod_{i=1}^r\max\{1,|\Sens(f_{\pi(x)},x^{(i)})|\}
    $$
    where $\pi(x)$ is the path followed by $x$ in $T$ and $f_{\pi(x)}$ is the restriction of $f$ by $\pi(x)$.
\end{lemma}

% \lnote{Say something about how these lemmas allow us to prove an XORed version of~\Cref{claim:slight-inapprox} by appling them with $f = \IsEdge$, and hwo the XORed version gives $s$ vs $2^r\cdot s$.}

Just as how~\Cref{lem:patchup-intro,lem:hard-distribution-intro} combine to yield~\Cref{claim:slight-inapprox}, combining their XOR-ed generalizations~\Cref{lem:patchup-intro,lem:hard-distribution-lemma-xor} yields the following amplified version of of~\Cref{claim:slight-inapprox}:

\violet{

\begin{claim}
\label{claim:strong-inapproximability}
    Let $G$ be a graph on $n$ vertices and $m$ edges. For all $\ell,r\ge 1$ and $\eps>0$, the following two cases hold.
    \begin{itemize}
        \item Yes case: if $G$ has a vertex cover of size $k$, then there is a decision tree $T$ computing $\ell\text{-}\mathrm{\IsEdge}_G^{\oplus r}$ whose size satisfies
        $$
        |T|\le \big[(\ell+1)(k+m)+mn\big]^r.
        $$
        \item No case: there is a distribution $\mathcal{D}$ such that if every vertex cover of $G$ has size at least $k'$, then any decision tree $T$ that is $\eps$-close to $\ell\text{-}\mathrm{\IsEdge}_G^{\oplus r}$ over $\mathcal{D}$ has size at least
        $$
        |T|\ge \big[(\ell+1)(k'+m)\big]^r-\eps \big[8m(\ell+1)\big]^r.
        $$
    \end{itemize}
\end{claim}

With an appropriate choice of parameters,~\Cref{claim:strong-inapproximability} translates a gap of $\le k$ versus $\ge (1+\delta)k$ in the vertex cover complexity of $G$ into a gap of $\le s^r$ versus $\ge (1+\delta')^r s^r$ in the decision tree complexity of $\ell$-$\IsEdge_G^{\oplus r}$, where $\delta'$ is a constant that depends only on $\delta$. \Cref{thm:main-general} follows as a corollary. See \Cref{fig:amplification} for an illustration of this amplification and how it fits into our overall reduction from {\sc Vertex Cover}.

%Recall that the reduction in \Cref{claim:slight-inapprox} translate a gap of $k$ versus $(1+c)k$ for {\sc Vertex Cover} into a gap of $s$ versus $(1+c')s$ in the size of decision trees for $\ell\text{-}\mathrm{\IsEdge}_G$. Using \Cref{lem:patchup-xor-intro,lem:hard-distribution-intro} with $f=\ell\text{-}\mathrm{\IsEdge}_G$, we can amplify this to a gap of $s^r$ versus $(1+c')^r s^r$ for all $r\ge 1$. Specifically, we prove the following claim which shows that the decision tree complexity of $\ell\text{-}\mathrm{\IsEdge}_G^{\oplus r}$ amplifies the gap in {\sc Vertex Cover} to the desired size gap of $s^r$ versus $(1+c')^r s^r$.

%Indeed, by choosing $\ell=\Theta(n)$ and $\eps=2^{-\Theta(r)}$ appropriately, if $s=\big[(\ell+1)(k+m)+mn\big]^r$ and $k'=(1+c)k$, then the No case of \Cref{claim:strong-inapproximability} can be written as $|T|\ge (1+c')^rs^r$ where $c'>0$ is an appropriately chosen constant. See \Cref{lem:technical-xor} for more details.
}

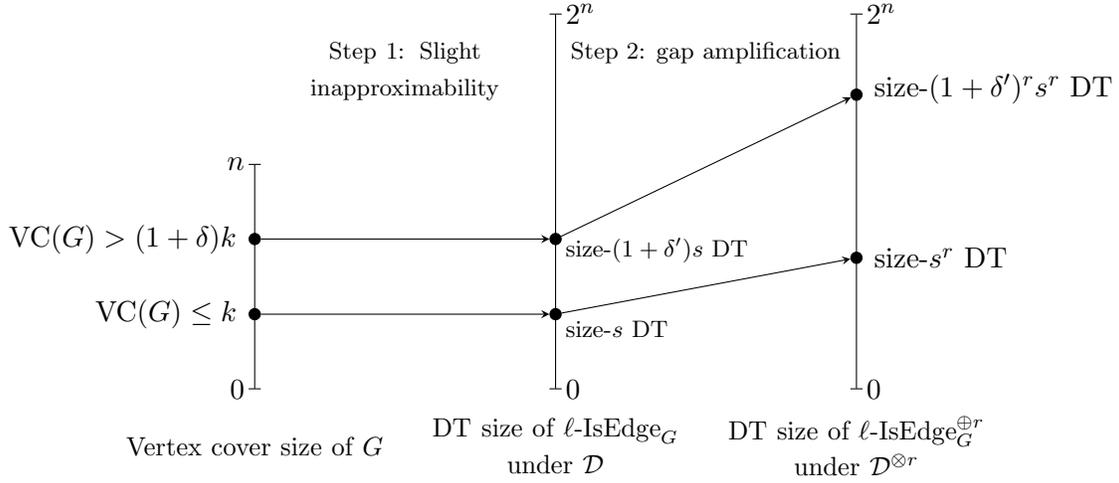
\begin{figure}[h!]
    \centering
    \begin{tikzpicture}[]
        % x DEFINES HOW WIDE FIGURE IS
        \def\x{4.0}
        % c DEFINES CENTER
        \def\c{0}
        % l DEFINES LEFT CENTER
        \def\l{-0.1}
        % l DEFINES RIGHT CENTER
        \def\r{0.1}
        % MAIN VERTICAL AXES
        \draw[black,xshift=-\x cm,|-|] (\c,-3) node[left] {$0$} -- (\c,0) node[left]{$n$};
        \draw[black, |-|] (\c,-3) node[right] {$0$} -- (\c,2) node[right]{$2^n$};
        \draw[black,xshift=\x cm, |-|] (\c,-3) node[right] {$0$} -- (\c,2) node[right]{$2^n$};
        % k AND ITS POINT
        \draw[color=black,xshift=-\x cm] (\l,-2) node[left,fill=white] {{\color{black}$\VC(G)\le k$}};
        \node[draw,circle,fill=black,inner sep=1.5pt,xshift=-\x cm] (k) at (\c,-2) {};
        % k' AND ITS POINT
        \draw[color=black,xshift=-\x cm] (\l,-1) node[left,fill=white] {\color{black} $\VC(G)>(1+\delta)k$};
        \node[draw,circle,fill=black,inner sep=1.5pt,xshift=-\x cm] (k') at (\c,-1) {};
        % k AND k' MIDDLE
        \node[draw,circle,fill=black,inner sep=1.5pt,xshift=0cm] (k'M) at (\c,-1) {};
        \node[draw,circle,fill=black,inner sep=1.5pt,xshift=0cm] (kM) at (\c,-2) {};
        % DT-SIZE LABELS
        \draw[] ([yshift=.15cm,xshift=0cm]k'M) node[anchor=north west] {\footnotesize{size-$(1+\delta')s$ DT}};
        \draw[] ([yshift=.05cm,xshift=0cm]kM) node[anchor=north west] {\footnotesize{size-$s$ DT}};

        % 2^k'l AND ITS POINT
        \draw[color=black,xshift=\x cm] (\r,1) node[right,fill=white] {size-$(1+\delta')^rs^r$ DT};
        \node[draw,circle,fill=black,inner sep=1.5pt,xshift=\x cm] (k'l) at (\c,0.92) {};
        % 2^kl AND ITS POINT
        \draw[color=black,xshift=\x cm] (\r,-1.25) node[right,fill=white] {size-$s^r$ DT};
        \node[draw,circle,fill=black,inner sep=1.5pt,xshift=\x cm] (kl) at (\c,-1.25) {};
        % EASY AND GAP AMPLIFICATION LABELS
        \draw[xshift=\x*0.5 cm] (\c,1.75) node[below,fill=white] {\footnotesize Step 2: gap amplification};
        \draw[xshift=-\x*0.5 cm] (\c,1.75) node[below,fill=white,text width=3cm,text centered] {\footnotesize Step 1: Slight\\ inapproximability};
        % LINES
        \draw[black,-stealth] (kM) to node[midway,below,sloped] {} (kl);
        \draw[black,-stealth] (k'M) to node[midway,above,sloped] {} (k'l);
        \draw[black,-stealth] (k') to node[midway,above,sloped] {} (k'M);
        \draw[black,-stealth] (k) to node[midway,below,sloped] {} (kM);
        % VERTICAL AXIS LABELS
        \draw[xshift=-\x cm] (\c,-3.75) node [fill=white,text width=3.5cm,text centered] {\small Vertex cover size of $G$};
        \draw[xshift=0cm] (\c,-3.75) node [fill=white,text width=3.9cm,text centered] {\small DT size of $\ell\text{-}\IsEdge_G$\\ under $\mathcal{D}$};
        \draw[xshift=\x cm] (\c,-3.75) node [fill=white,text width=3.9cm,text centered] {\small DT size of $\ell\text{-}\IsEdge_G^{\oplus r}$ under $\mathcal{D}^{\otimes r}$};
    \end{tikzpicture}
    \caption{An illustration of main reduction from {\sc Vertex Cover} in two steps. The first step, which establishes slight inapproximability of decision tree learning, is proved in \Cref{claim:slight-inapprox}. The second step amplifies this slight inapproximability gap using \Cref{claim:strong-inapproximability}. 
    % In the first step, proved in \Cref{claim:slight-inapprox}, we reduce the problem of deciding whether a graph has a vertex cover of size $\le k$ or whether every vertex cover has size $>(1+\delta)k$ for some constant $\delta>0$ to the problem of deciding whether $\ell\text{-}\IsEdge_G$ has a decision tree of size $s$ or requires size $(1+\delta')s$ for some constant $\delta'>0$ to approximate under $\mathcal{D}$. In the second step, we amplify the gap of $s$ versus $(1+\delta')s$ to $s^r$ versus $(1+\delta')^rs^r$ by taking the XOR of $\ell\text{-}\IsEdge_G$. This step is proved in \Cref{claim:strong-inapproximability}. 
    }
    \label{fig:amplification}
\end{figure}
\section{Preliminaries}
\label{sec:prelims} 
\paragraph{Notation and naming conventions.}{We write $[n]$ to denote the set $\{1,2,\ldots,n\}$. We use lower case letters to denote bitstrings e.g. $x,y\in\zo^n$ and subscripts to denote bit indices: $x_i$ for $i\in [n]$ is the $i$th index of $x$. The string $x^{\oplus i}$ is $x$ with its $i$th bit flipped. We use superscripts to denote multiple bitstrings of the same dimension, e.g. $x^{(1)},x^{(2)},...,x^{(j)}\in\zo^n$. For a set $S$ and an integer $r\ge 1$, we write $S^r$ to denote the $r$-ary Cartesian product of the set. 
}

\paragraph{Distributions.}{We use boldface letters e.g. $\bx,\by$ to denote random variables. For a distribution $\mathcal{D}$, we write $\error_{\mathcal{D}}(f,g)=\Pr_{\bx\sim \mathcal{D}}[f(\bx)\neq g(\bx)]$. The support of the distribution is the set of elements with nonzero mass and is denoted $\supp(\mathcal{D})$. For $r\ge 1$, we write $\mathcal{D}^{\otimes r}$ to denote the $r$-wise product distribution $\underbrace{\mathcal{D}\times\cdots\times \mathcal{D}}_{r\text{ times }}$.
}

\paragraph{{Decision trees.}}{ The size of a decision tree $T$ is its number of leaves and is denoted $|T|$. In an abuse of notation, we also write $T$ for the function computed by the decision tree $T$. We say $T$ computes a function $f:\zo^n\to\zo$ if $T(x)=f(x)$ for all $x\in\zo^n$. The decision tree complexity of a function $f$ is the size of the smallest decision tree computing $f$ and is denoted $\dtsize(f)$.}

\paragraph{Restrictions and decision tree paths.}{
A restriction $\rho$ is a set $\rho\sse \{x_1,\overline{x}_1,\ldots,x_n,\overline{x}_n\}$ of literals, and $f_\rho$ is the subfunction obtained by restricting $f$ according to $\rho$: $f_\rho(x^\star)=f(x^\star\vert_\rho)$ where $x^\star\vert_\rho$ is the string obtained from $x^\star$ by setting its $i$th coordinate to $1$ if $x_i\in \rho$, $0$ if $\overline{x}_i\in \rho$, and otherwise setting it to $x^\star_i$. We say an input $x^\star$ is consistent with $\rho$ if $x_i\in \rho$ implies $x^\star_i=1$ and $\overline{x}_i\in \rho$ implies $x^\star_i=0$. 

Paths in decision trees naturally correspond to restrictions. A depth-$d$ path can be identified with a set of $d$ literals: $\pi=\{\ell_1,\ell_2,\ldots,\ell_d\}$ where each $\ell_i$ corresponds to a query of an input variable and is unnegated if $\pi$ follows the right branch and negated if $\pi$ follows the left branch. 
}

\paragraph{Boolean Functions.} \label{para:boolean-functions}
We use $f$ to denote an arbitrary $n$-bit Boolean function, $f:\zo^n\to\zo$. For a set $D\sse \zo^n$, we write $f:D\to\zo$ for the partial Boolean function defined on $D$. We use both partial and total functions and specify the setting by writing either $f:\zo^n\to\zo$ or $f:D\to\zo$. For $f:D\to\zo$, the sensitivity of $f$ on $x\in D$ and the certificate complexity of $f$'s value on $x$ are defined as 
\begin{align*}
    \Sens(f,x)&=\{x^{\oplus i}\in D: f(x)\neq f(x^{\oplus i})\text{ for }i\in [n]\}\\
    \Cert(f,x)&=|\pi|\text{ s.t. }\pi \text{ is the shortest restriction consistent with }x\text{ and }\\
    &\quad\quad f_\pi\text{ is a constant function}.
\end{align*}
Note that both of these definitions are with respect to $D$. Also, we refer to the \textit{sensitivity of $f$} which is $\Sens(f)\coloneqq\max_{x\in D}|\Sens(f,x)|$. 

\paragraph{Graphs.}{An undirected graph $G=(V,E)$ has $n$ vertices $V\sse [n]$ and $m=|E|$ edges $E\sse V\times V$. The degree of a vertex $v\in V$ is the number of edges containing it: $|\{e\in E: v\in e\}|$. The graph $G$ is degree-$d$ if every vertex $v\in V$ has degree at most $d$. We often use letters $v,u,w$ to denote vertices of a graph $G$. 
}

\paragraph{Learning.}{In the PAC learning model, there is an unknown distribution $\mathcal{D}$ and some unknown \textit{target} function $f\in\mathcal{C}$ from a fixed \textit{concept} class $\mathcal{C}$ of functions over a fixed domain. An algorithm for learning $\mathcal{C}$ over $\mathcal{D}$ takes as input an error parameter $\eps\in (0,1)$ and has oracle access to an \textit{example oracle} $\textnormal{EX}(f,\mathcal{D})$. The algorithm can query the example oracle to receive a pair $(\bx,f(\bx))$ where $\bx\sim\mathcal{D}$ is drawn independently at random. The goal is to output a \textit{hypothesis} $h$ such that $\dist_{\mathcal{D}}(f,h)\le \eps$. Since the example oracle is inherently randomized, any learning algorithm is necessarily randomized. So we require the learner to succeed with some fixed probability e.g.~$2/3$. A learning algorithm is \textit{proper} if it always outputs a hypothesis $h\in\mathcal{C}$. A learning algorithm with \textit{queries} is given oracle access to the target function $f$ along with the example oracle $\textnormal{EX}(f,\mathcal{D})$. }

\subsection{Vertex Cover} 

\paragraph{Vertex cover.}{A vertex cover for an undirected graph $G=(V,E)$ is a subset of the vertices $C\sse V$ such that every edge has at least one endpoint in $C$. We write $\VC(G)$ to denote the size of the smallest vertex cover of $G$. The {\sc Vertex Cover} problem is to decide whether a graph contains a vertex cover of size $\le k$. For $a>1$, an $a$-approximation of {\sc Vertex Cover} corresponds to the problem of deciding whether a graph contains a vertex cover of size $\le k$ or every vertex cover has size at least $a\cdot k$. There is a polynomial-time greedy algorithm for vertex cover which achieves a $2$-approximation. 

\paragraph{Hardness of approximation.} Constant factor hardness of \textsc{Vertex Cover} is known, even for bounded degree graphs (graphs whose degree is bounded by some universal constant). This is the main hardness result we reduce from in this work.

\begin{theorem}[Hardness of approximating \textsc{Vertex Cover}]
\label{thm:hardness of vertex cover} 
There are constants $\delta>0$ and $d\in \N$ such that if {\sc Vertex Cover} can be $(1+\delta)$-approximated on $n$-vertex degree-$d$ graphs in time $t(n)$, then \SAT\ can be solved in time $t(n\polylog n)$.
\end{theorem}
}

\Cref{thm:hardness of vertex cover} follows from the works of \cite{PY91,AS98, ALMSS98}. The fact that \Cref{thm:hardness of vertex cover} holds for constant degree graphs will be essential for our lower bound because it allows us to assume that $k$ is large: $\VC(G)=\Theta(m)$.

\begin{fact}[Constant degree graphs require large vertex covers]
    \label{fact:constant degree graphs have large vc size}
    If $G$ is an $m$-edge degree-$d$ graph, then $\VC(G)\ge m/d$.
\end{fact}

This fact follows from the observation that in a degree-$d$ graph each vertex can cover at most $d$ edges.

\input{body-of-paper}

% \section*{Acknowledgments}

% 

\acks{We thank the COLT reviewers for their helpful feedback and suggestions.

The authors are supported by NSF awards 1942123, 2211237, 2224246, a Sloan Research Fellowship, and a Google Research Scholar Award. Caleb is also supported by an NDSEG Fellowship, and Carmen by a Stanford Computer Science Distinguished Fellowship and an NSF Graduate Research Fellowship.}

 \bibliography{ref}

\newpage

\appendix
% \section{Patching up a decision tree: Proof of \Cref{lem:patchup-intro}}

% First, we prove the following claim about building a decision tree from scratch using certificates.

% \begin{claim}[Building a decision tree out of certificates]
% \label{claim:zero-error-patch-up}
%     Let $f:D\to\zo$ be a function with $D\sse \zo^n$, then
%     $$
%     \dtsize(f)\le 1+\sum_{x\in f^{-1}(1)}\Cert(f,x).
%     $$
% \end{claim}

\section{Proof of \Cref{claim:zero-error-patch-up}}
\label{subsec:zero-error-patch-up}
        Let $T$ be the decision tree built iteratively by the following procedure. In the first iteration, pick an arbitrary $x\in f^{-1}(1)$ and fully query the indices in $\Cert(f,x)$. Let $T_i$ be the tree formed after the $i$th iteration. Then $T_{i+1}$ is formed by choosing $x\in f^{-1}(1)$ which has not been picked in a previous iteration. Then, at the leaf reached by $x$ in $T_i$, fully query the indices in $\Cert(f,x)$ (ignoring those indices which have already been queried along the path followed by $x$ in $T_i$). Repeating this for $|f^{-1}(1)|$ steps, yields a decision tree with at most $\sum_{x\in f^{-1}(1)}\Cert(f,x)$ \textit{internal nodes}. Therefore, the number of leaves is at most $1+\sum_{x\in f^{-1}(1)}\Cert(f,x)$. 
        
        It remains to show that this tree exactly computes $f$. Specifically, we'll argue that $f_\pi$ is the constant function for every path $\pi$ in the decision tree. If not, then there is an $x\in f^{-1}(1)$ so that $x$ follows the path $\pi$ and $f_\pi$ is nonconstant. But this is a contradiction since $\pi$ consists of a certificate of $x$ by construction.\hfill $\blacksquare$

% The first lemma shows that we can patch up a decision tree by querying certificates. This recovers \Cref{lem:patchup-intro} in the setting where $D=\zo^n$.

% \begin{lemma}[Patch-up with respect to $1$-inputs]
% \label[lemma]{lem:patch-up-dt}
%     Let $f:D\to\zo$ be a function with $D\sse \zo^n$ and let $T$ be a decision tree, then
%     $$
%     \dtsize(f)\le |T|+\sum_{x\in f^{-1}(1)}\Cert(f_{\pi(x)},x)
%     $$
% \end{lemma}

% \begin{proof}
%     Let $\Pi$ denote the set of paths in $T$. Then,
%     \begin{align*}
%         \dtsize(f)&\le \sum_{\pi\in\Pi}\dtsize(f_\pi)\\
%         &\le \sum_{\pi\in\Pi}\left(1+\sum_{x\in f_\pi^{-1}(1)}\Cert(f_\pi,x)\right)\tag{\Cref{claim:zero-error-patch-up}}\\
%         &= |T|+\sum_{\pi\in\Pi}\sum_{x\in f_\pi^{-1}(1)}\Cert(f_\pi,x)\tag{$|\Pi|=|T|$}\\
%         &=|T|+\sum_{x\in f^{-1}(1)}\Cert(f_{\pi(x)},x).
%     \end{align*}
%     The last equality follows from the fact that the set of paths $\pi\in \Pi$ partition $f^{-1}(1)$.
% \end{proof}

\section{Hard distribution lemma: Proof of \Cref{lem:hard-distribution-lemma}}
\label{sec:hard-dist-lemma}
% \begin{definition}[Canonical hard distribution]
% \label{def:canonical-hard-distribution}
%     For a function $f:D\to\zo$ with $D\sse \zo^n$, the canonical hard distribution, $\mathcal{D}_f$, is defined via the following experiment
%     \begin{itemize}
%         \item sample $\bx\sim f^{-1}(1)$ u.a.r.
%         \item with probability $1/2$, return $\by\sim\Sens(f,\bx)$ u.a.r.
%         \item with probability $1/2$, return $\bx$.
%     \end{itemize}
% \end{definition} 
% When $f$ is clear from context, we simply write $\mathcal{D}$.

% We use the canonical hard distribution to prove the following result.
% \begin{lemma}[Hard distribution lemma]
% \label[lemma]{lem:hard-distribution-lemma}
%     Let $f:D\to\zo$ be a nonconstant function for $D\sse \zo^n$. Then for all $C\sse f^{-1}(1)$, there exists a distribution $\mathcal{D}$ over $C$ and and all of its sensitive neighbors such that for any decision tree $T$, we have
%     $$
%     \mathrm{error}_{\mathcal{D}}(T,f)\ge \frac{1}{2|C|\Sens(f)}\sum_{x\in C}|\Sens(f_{\pi(x)},x)|.
%     $$
% \end{lemma}

First we prove a lemma which counts the error under $\mathcal{D}$ conditioned on the $1$-input obtained in the first sampling step. 
\begin{lemma}[Error with respect to the canonical hard distribution conditioned on a $1$-input]
\label[lemma]{lem:error-conditioned-on-one}
    Let $T$ be a decision tree, $f:D\to\zo$ be a function, and let $\mathcal{D}$ be the canonical hard distribution. For all $x\in f^{-1}(1)$, 
    $$
    \Prx_{\bz\sim\mathcal{D}_f}[T(\bz)\neq f(\bz)\mid \text{ first sampling }x\text{ from }f^{-1}(1)]\ge \frac{1}{2}\cdot\frac{|\Sens(f_{\pi(x)},x)|}{\max\{1,|\Sens(f,x)|\}}.
    $$
\end{lemma}

\begin{proof}
    If $|\Sens(f,x)|=0$ then $|\Sens(f_{\pi(x)},x)|=0$. And if the RHS is $0$ then the bound is vacuously true. Assume that $\Sens(f_{\pi(x)},x)\neq \varnothing$. Both $x$ and all $y\in \Sens(f_{\pi(x)},x)$ follow the same path in $T$ and have the same leaf label. Since $f(x)=1$ and $f(y)=0$, we can write
    \begin{align*}
        \Prx_{\bz\sim\mathcal{D}_f}\big[T(\bz)\neq f(\bz)\mid \textnormal{first sampling }x&\textnormal{ from }f^{-1}(1)\big]\\
        \ge \min\bigg\{&\Prx_{\bz\sim\mathcal{D}_f}\big[\bz=x\mid \textnormal{first sampling }x\textnormal{ from }f^{-1}(1)\big],\\
        &\Prx_{\bz\sim\mathcal{D}_f}\big[\bz\in\Sens(f_{\pi(x)},x)\mid \textnormal{first sampling }x\textnormal{ from }f^{-1}(1)\big]\bigg\}\\
        =\min\bigg\{&\frac{1}{2},\frac{1}{2}\cdot \frac{|\Sens(f_{\pi(x)},x)|}{|\Sens(f,x)|}\bigg\}\ge \frac{1}{2}\cdot\frac{|\Sens(f_{\pi(x)},x)|}{|\Sens(f,x)|}
        % \ge\min\bigg\{&1-p,\frac{p}{\Sens(f)}\bigg\}\tag{$|\Sens(f_{\pi(x)},x)|\ge 1$ and $|\Sens(f,x)|\le \Sens(f)$}
    \end{align*}
    where the inequality follows from the fact that the probability of an error is lower bounded by the probability that $\bz=x$ when the label for the path in $T$ is $0$ and is lower bounded by the probability that $\bz\in\Sens(f_{\pi(x)},x)$ when the label for the path is $1$. Note that if $\Sens(f_{\pi(x)},x)\neq \varnothing$ then necessarily $|\Sens(f,x)|\ge 1$ and $\max\{1,|\Sens(f,x)|\}=|\Sens(f,x)|$. Therefore, the proof is complete.
\end{proof}

\begin{proof}[Proof of \Cref{lem:hard-distribution-lemma}]
    We prove the statement for $C=f^{-1}(1)$. If $C\neq f^{-1}(1)$, we can consider the function $f:C\cup C'\to \zo$ where $C'$ denotes the set of sensitive neighbors of strings in $C$, and the same proof holds. Let $\mathcal{D}$ denote the distribution from \Cref{def:canonical-hard-distribution} and notice that the support of this distribution is $C$ and all of its sensitive neighbors. We have
    \begin{align*}
        \mathrm{error}_{\mathcal{D}}(T,f)&=\Prx_{\bz\sim\mathcal{D}}[T(\bz)\neq f(\bz)]\\
        &=\sum_{x\in C}\frac{1}{|C|}\cdot \Prx_{\bz\sim\mathcal{D}}[T(\bz)\neq f(\bz)\mid \text{ first sampling }x\text{ from }f^{-1}(1)]\\
        &\ge \frac{1}{2|C|}\sum_{x\in C}\frac{|\Sens(f_{\pi(x)},x)|}{\max\{1,|\Sens(f,x)|\}}\tag{\Cref{lem:error-conditioned-on-one}}\\
        &\ge \frac{1}{2|C|\cdot\Sens(f)}\sum_{x\in C}|\Sens(f_{\pi(x)},x)|\tag{$\max\{1,|\Sens(f,x)|\}\le \Sens(f)$ for all $x$}
    \end{align*}
    where the last step uses the fact that $\Sens(f)\ge 1$ since $f$ is nonconstant. In particular, we have that $\max\{1,|\Sens(f,x)|\}\le \Sens(f)$ for all $x$.
\end{proof}

\section{Deferred proofs for \Cref{thm:KST}}

\subsection{Proof of \Cref{lem:ell-isedge-lb}}
\label{subsec:ell-isedge-lb}
First, we give a basic zero-error lower bound on $\ell\text{-}\IsEdge$ and observe some properties about the sensitivity and certificate complexity of $\ell\text{-}\IsEdge$ in \Cref{lem:ell-isedge-lb-zero-error,prop:sens-ell-isedge,prop:sens-equals-cert}, respectively.

\begin{lemma}[Zero-error lower bound for $\ell\text{-}\IsEdge_G$; see {\cite[Claim 6.6]{KST23}}]
\label[lemma]{lem:ell-isedge-lb-zero-error}
    Let $G$ be an $n$-vertex, $m$-edge graph where every vertex cover has size at least $k'$. Then, any decision tree $T$ computing $\ell\text{-}\IsEdge_G:\ell\text{-}{D}_G\to\zo$ over $\ell\text{-}{D}_G$, the support of the canonical hard distribution, must have size
    $$
    |T|\ge (\ell+1)(k'+m).
    $$
\end{lemma}

\begin{proof}
    The same lower bound is proved in Claim 6.6 of \cite{KST23} under a slightly different subset of inputs. Specifically, they prove $\dtsize(\ell\text{-}\IsEdge_G)\ge (\ell+1)(k'+m)$ where $\ell\text{-}\IsEdge_G:D'\to\zo$ for the set $D'=\ell\text{-}{D}_G\cup\{0^{n+\ell n}\}$ which adds the all 0s input. This small difference doesn't change the lower bound since any decision tree $T$ computing $\ell\text{-}\IsEdge_G$ over the set $\ell\text{-}{D}_G$ also computes it over $D'$. Indeed, every $1$-input to $\ell\text{-}\IsEdge_G$ is sensitive on every $1$-coordinate and so if $T$ satisfies $T(x)=\ell\text{-}\IsEdge_G(x)$ for every $x\in \ell\text{-}{D}_G$, then it must query every $1$-coordinate of each $1$-input. Therefore, we can assume without loss of generality that $T(0^{n+\ell n})=0$.
\end{proof}

\begin{proposition}[Sensitivity of $\ell\text{-}\IsEdge_G$]
\label{prop:sens-ell-isedge}
    For a graph $G$, $\ell\ge 1$, and $\ell\text{-}\IsEdge_G:\ell\text{-}{D}_G\to\zo$, we have
    $$
    \Sens(\ell\text{-}\IsEdge_G)=2(\ell+1).
    $$
\end{proposition}

\begin{proof}
    Let $\ell\text{-}\ind[e]$ be an edge indicator for an edge $e\in E$. Let $i\in [n\ell+n]$ denote the index of a $1$-coordinate of $\ell\text{-}\ind[e]$. By definition, there are $2(\ell+1)$ many such $i$ and each $i$ is sensitive: $\ell\text{-}\IsEdge_G(\ell\text{-}\ind[e]^{\oplus i})=0$. Therefore, $|\Sens(\ell\text{-}\IsEdge_G,\ell\text{-}\ind[e])|=2(\ell+1)$. Conversely, for every sensitive neighbor $\ell\text{-}\ind[e]^{\oplus i}$, we have $\Sens(\ell\text{-}\IsEdge_G,\ell\text{-}\ind[e]^{\oplus i})=\{\ell\text{-}\ind[e]\}$ and so $|\Sens(\ell\text{-}\IsEdge_G,\ell\text{-}\ind[e]^{\oplus i})|=1$. Thus the overall sensitivity is $\Sens(\ell\text{-}\IsEdge_G)=2(\ell+1).$
\end{proof}

\begin{proposition}[Sensitivity equals certificate complexity of $1$-inputs]
    \label[proposition]{prop:sens-equals-cert}
    Let $G$ be a graph and $\ell\text{-}\IsEdge:\ell\text{-}{D}_G\to\zo$, the corresponding edge function. For all edge indicators $x=\ell\text{-}\ind[e]$ and for all restrictions $\pi$, we have
    $$
    \Cert(\ell\text{-}\IsEdge_\pi,x)=|\Sens(\ell\text{-}\IsEdge_\pi,x)|.
    $$
\end{proposition}

\begin{proof}
    By definition $|\Sens(\ell\text{-}\IsEdge_\pi,x)|$ is the number of $1$-coordinates in $x$ which are not restricted by $\pi$. The set of $1$-coordinates of $x$ not restricted by $\pi$ forms a certificate of $\ell\text{-}\IsEdge_\pi$ since fixing these coordinates forces $\ell\text{-}\IsEdge$ to be the constant $1$-function. It follows that $\Cert(\ell\text{-}\IsEdge_\pi,x)=|\Sens(\ell\text{-}\IsEdge_\pi,x)|$.
\end{proof}

% \begin{claim}
% \label{claim:isedge-lb-sens}
%     Let $T$ be a decision tree for $\ell\text{-}\IsEdge$ with error $\eps$ over the canonical hard distribution $\ell\text{-}\mathcal{D}_G$. Then,
%     $$
%     |T|\ge \dtsize(\ell\text{-}\IsEdge,\ell\text{-}{D}_G)-2\eps m\cdot\Sens(\ell\text{-}\IsEdge)
%     $$
% \end{claim}

% \begin{proof}
%     For a graph consisting of $m$ edges, the number of $1$-inputs to $\ell\text{-}\IsEdge:\ell\text{-}{D}_G\to\zo$ is $m$. Therefore,
%     \begin{align*}
%         \eps &\ge \frac{1}{2m\cdot\Sens(\ell\text{-}\IsEdge)}\sum_{x\in \ell\text{-}\IsEdge^{-1}(1)}|\Sens(\ell\text{-}\IsEdge_{\pi(x)},x)|\tag{\Cref{thm:hard-distribution-lemma}}\\
%         &=\frac{1}{2m\cdot \Sens(\ell\text{-}\IsEdge)}\sum_{x\in \ell\text{-}\IsEdge^{-1}(1)}\Cert(\ell\text{-}\IsEdge_{\pi(x)},x)\tag{\Cref{claim:sens-equals-cert}}\\
%         &\ge \frac{1}{2m\cdot \Sens(\ell\text{-}\IsEdge)}\left(\dtsize(\ell\text{-}\IsEdge,\ell\text{-}{D}_G)-|T|\right)\tag{\Cref{lem:patch-up-dt}}
%     \end{align*}
%     and the proof is completed by rearranging the above inequality.
% \end{proof}

\begin{proof}[Proof of \Cref{lem:ell-isedge-lb}]
    For a graph consisting of $m$ edges, the number of $1$-inputs to $\ell\text{-}\IsEdge:\ell\text{-}{D}_G\to\zo$ is $m$. Therefore,
    \begin{align*}
        \eps &\ge \frac{1}{2m\cdot\Sens(\ell\text{-}\IsEdge)}\sum_{x\in \ell\text{-}\IsEdge^{-1}(1)}|\Sens(\ell\text{-}\IsEdge_{\pi(x)},x)|\tag{\Cref{lem:hard-distribution-lemma}}\\
        &=\frac{1}{2m\cdot \Sens(\ell\text{-}\IsEdge)}\sum_{x\in \ell\text{-}\IsEdge^{-1}(1)}\Cert(\ell\text{-}\IsEdge_{\pi(x)},x)\tag{\Cref{prop:sens-equals-cert}}\\
        &\ge \frac{1}{2m\cdot \Sens(\ell\text{-}\IsEdge)}\left(\dtsize(\ell\text{-}\IsEdge,\ell\text{-}{D}_G)-|T|\right).\tag{\Cref{lem:patch-up-dt}}
    \end{align*}
    Rearranging the above, we obtain
    \begin{align*}
    |T|&\ge \dtsize(\ell\text{-}\IsEdge,\ell\text{-}{D}_G)-2\eps m\cdot\Sens(\ell\text{-}\IsEdge)\\
    &\ge(\ell+1)(k'+m)-2\eps m\cdot\Sens(\ell\text{-}\IsEdge)\tag{\Cref{lem:ell-isedge-lb-zero-error}}\\
    &=(\ell+1)(k'+m)-4\eps m(\ell+1)\tag{\Cref{prop:sens-ell-isedge}}
    \end{align*}
    which completes the proof.
    % Any decision tree $T$ with error $\eps$ over the canonical hard distribution satisfies
    % \begin{align*}
    % |T|&\ge \dtsize(\ell\text{-}\IsEdge_G,\ell\text{-}{D}_G)-2\eps m\cdot\Sens(\ell\text{-}\IsEdge_G)\tag{\Cref{thm:patch-up-monotone}}\\
    % &\ge(\ell+1)(k'+m)-2\eps m\cdot\Sens(\ell\text{-}\IsEdge_G)\tag{\Cref{lem:ell-isedge-lb-zero-error}}\\
    % &=(\ell+1)(k'+m)-4\eps m(\ell+1)\tag{\Cref{prop:sens-ell-isedge}}
    % \end{align*}
    % which completes the proof.
\end{proof}

\subsection{Proof of \Cref{lem:technical-kst}}
\label{subsec:technical-kst}
% The following key lemma is used in the analysis of correctness for the reduction in \Cref{thm:KST}.

% \begin{lemma}[Main technical lemma]
% \label[lemma]{lem:technical-kst}
% For all $\delta,\delta',\eps >0$ and $d,k\ge 1$, the following holds. Given a constant degree-$d$ graph $G$ with $m$ edges and parameter $k$, there is a choice of $\ell=\Theta(|G|)$ and a polynomial-time computable quantity $s\in\N$ such that so long as $\delta'>(\delta+4\eps)d+\delta$ and $dk\ge m$ we have:
% \begin{itemize}
%     \item \textbf{Yes case}: if $G$ has a vertex cover of size at most $k$, then there is a decision tree of size at most $s$ which computes $\ell\text{-}\IsEdge:\zo^{n\ell+n}\to\zo$; and
%     \item \textbf{No case}: if every vertex cover of $G$ has size at least $(1+\delta')k$, then $(1+\delta)s<|T|$ for any decision tree $T$ with $\error_{\ell\text{-}\mathcal{D}_G}(T,\ell\text{-}\IsEdge)\le \eps$.
% \end{itemize}
% \end{lemma}

This lemma is a consequence of the following proposition along with the upper and lower bounds we have obtained for $\ell\text{-}\IsEdge$. The proposition is a calculation involving the parameters that come into play in \Cref{lem:technical-kst}. We state it on its own, since we will reuse the calculation later when proving \Cref{thm:main-general}.

\begin{proposition}
\label[proposition]{prop:inequality}
    For all $\delta,\delta',\alpha>0$ and $\ell,m,n,k,d\ge 1$ satisfying $m\le dk$ and $\delta'>(\delta+\alpha)d+\delta+\frac{(1+\delta)mn}{k(\ell+1)}$, we have
    $$
    (1+\delta)\left[(\ell+1)(k+m)+mn\right]<(\ell+1)\left[(1+\delta')k+(1-\alpha) m\right].
    $$
\end{proposition}

\begin{proof}
    The proof is a calculation. We can write
    \begin{align*}
        (1+\delta')k-(1+\delta)k&=(\delta'-\delta)k\\
        &>\left(\delta+(\delta+\alpha)d+\frac{(1+\delta)mn}{k(\ell+1)}-\delta\right)k\tag{Assumption on $\delta'$}\\
        &=(\delta+\alpha)dk+\frac{(1+\delta)mn}{\ell+1}\\
        &\ge (\delta+\alpha)m+\frac{(1+\delta)mn}{\ell+1}\tag{$m\le dk$}\\
        &=(1+\delta)m-(1-\alpha)m+\frac{(1+\delta)mn}{\ell+1}.
    \end{align*}
    Therefore, rearranging we obtain
    $$
    (1+\delta)\left[(\ell+1)(k+m)+mn\right]<(\ell+1)\left[(1+\delta')k+(1-\alpha) m\right]
    $$
    which completes the proof. 
\end{proof}

\begin{proof}[Proof of \Cref{lem:technical-kst}]
    Given a degree-$d$, $m$-edge, $n$-vertex graph $G$ and parameter $k$, we choose $\ell=\Theta(n)$ so that $\delta'>(\delta+4\eps)d+\delta+\frac{(1+\delta)mn}{k(\ell+1)}$ and set $s=(\ell+1)(k+m)+mn$. Note that such an $\ell$ exists since $k=\Theta(n)$ for constant-degree graphs. We now prove the two points separately.

    \paragraph{Yes case.}{
        In this case, we have by \Cref{thm:ell-isedge-ub} that there is a decision tree $T$ computing $\ell\text{-}\IsEdge:\zo^{n\ell+n}\to\zo$ whose size satisfies
        $$
        |T|\le (\ell+1)(k+m)+mn=s.
        $$
    }
    \paragraph{No case.}{
        Let $T$ be a decision tree satisfying $\error_{\ell\text{-}\mathcal{D}_G}(T,\ell\text{-}\IsEdge)\le \eps$. Then, using our assumptions on the parameters: 
        \begin{align*}
        (1+\delta)s&=(1+\delta)\left[(\ell+1)(k+m)+mn\right]\tag{Definition of $s$}\\
            &<(\ell+1)\left[(1+\delta')k+(1-4\eps) m\right]\tag{\Cref{prop:inequality} with $\alpha=4\eps$}\\
            &\le|T|\tag{\Cref{lem:ell-isedge-lb} with $k'=(1+\delta')k$}.
        \end{align*}
    }
    We've shown the desired bounds in both the Yes and No cases so the proof is complete.
\end{proof}

\subsection{Proof of \Cref{thm:KST}}
\label{subsec:KST}

% \begin{proof}[Proof of \Cref{thm:KST}]
Let $G$ be a constant degree-$d$, $n$-vertex graph and $k\in \N$, a parameter. Let $\mathcal{A}$ be the algorithm for {\sc DT-Learn} from the theorem statement. We'll use $\mathcal{A}$ to approximate {\sc Vertex Cover} on $G$.

\paragraph{The reduction.}{
First, we check whether $dk\ge m$. If $dk<m$, our algorithm outputs ``No'' as $G$ cannot have a vertex cover of size at most $k$. Otherwise, we proceed under the assumption that $dk\ge m$. Let $s\in \N$ be the quantity from \Cref{lem:technical-kst}. We will run $\mathcal{A}$ over the distribution $\ell\text{-}\mathcal{D}_G$ and on the function $\ell\text{-}\IsEdge_G:\zo^{N}\to\zo$ where $\ell$ is as in \Cref{lem:technical-kst}. Note that $N=n\ell+n=O(n^2)$ and $s=O(n^2)=O(N)$. See \Cref{fig:solving vc with dt learn} for the exact procedure we run. 
}

\begin{figure}[h!]
\begin{tcolorbox}[colback = white,arc=1mm, boxrule=0.25mm]
\vspace{3pt}
$\textsc{Vertex Cover}(k,(1+\delta')\cdot k)$:
\begin{itemize}[leftmargin=20pt,align=left]
\item[\textbf{Given:}] $G$, an $m$-edge degree-$d$ graph over $n$ vertices and $k\in \N$
\item[\textbf{Run:}]\textsc{DT-Learn}$(s,(1+\delta)\cdot s,\eps)$ for $t(N,1/\eps)$ time steps providing the learner with
\begin{itemize}[align=left,labelsep*=0pt]
    \item \textit{queries}: return $\ell\text{-}{\mathrm{\IsEdge}}(v^{(0)},\ldots,v^{(\ell)})$ for a query $(v^{(0)},\ldots,v^{(\ell)})\in\zo^N$; and
    \item \textit{random samples}: return $(\bv^{(0)},\ldots,\bv^{(\ell)})\sim\ell\text{-}\mathcal{D}_G$ for a random sample.
\end{itemize}
\item[$T_{\text{hyp}}\leftarrow$ decision tree output of the learner]
\item[$\eps_{\text{hyp}}\leftarrow \mathrm{error}_{\ell\text{-}\mathcal{D}_G}(T_{\text{hyp}},\ell\text{-}{\mathrm{\IsEdge}})$]
\item[\textbf{Output:}] \textsc{Yes} if and only if $|T_{\text{hyp}}|\le (1+\delta)\cdot s$ and $\eps_{\text{hyp}}\le \eps$
\end{itemize}
\vspace{3pt}
\end{tcolorbox}
\medskip
\caption{Using an algorithm for \textsc{DT-Learn} to solve \textsc{Vertex Cover}.}
\label{fig:solving vc with dt learn}
\end{figure}

\paragraph{Runtime.}{
Any query to $\ell\text{-}{\mathrm{\IsEdge}}$ can be answered in $O(N)$ time. Similarly, a random sample can be obtained in $O(N)$ time. The algorithm uses $O(N\cdot t(N,1/\eps))$ time to run \textsc{DT-Learn}. Finally, computing $\mathrm{error}_{\ell\text{-}\mathcal{D}_G}(T_{\text{hyp}},\ell\text{-}{\mathrm{\IsEdge}})$ takes $O(N^2)$. Since $t(N,1/\eps)\ge N$, the overall runtime is $O(N\cdot t(N,1/\eps))=O(n^2t(n^2,1/\eps))$. 
}

\paragraph{Correctness.}{
Correctness follows from \Cref{lem:technical-kst}. Specifically, in the \textbf{Yes case}, if $G$ has a vertex cover of size at most $k$, then there is a decision tree of size at most $s$ computing $\ell\text{-}{\mathrm{\IsEdge}}$.  Therefore, by the guarantees of {\sc DT-Learn}, we have $|T_{\text{hyp}}|\le (1+\delta)\cdot s$ and $\eps_{\text{hyp}}\le \eps$ and our algorithm correctly outputs ``Yes''.  

In the \textbf{No case}, every vertex cover of $G$ has size at least $(1+\delta')k$. If $\eps_{\text{hyp}}>\eps$ then our algorithm for {\sc Vertex Cover} correctly outputs ``No''. Otherwise, assume that $\eps_{\text{hyp}}\le \eps$. Then, \Cref{lem:technical-kst} ensures that $(1+\delta)s<|T_{\text{hyp}}|$ and so our algorithm correctly outputs ``No'' in this case as well.
}
This completes the proof.\hfill$\blacksquare$
% \end{proof}
\section{Proof of \Cref{lem:patchup-xor}}
\label{sec:patchup-xor}
% The following lemma recovers \Cref{lem:patchup-xor-intro} by setting $D=\zo^n$. 

% \begin{lemma}[XOR-ed version of patchup lemma, formal statement of \Cref{lem:patchup-xor-intro}]
% \label[lemma]{lem:patchup-xor}
%     Let $T$ be a decision tree and $f:D\to\zo$ be a nonconstant function for $D\sse \zo^n$, then
%     $$
%     \dtsize(f^{\oplus r})\le |T|+2^{r}\sum_{\substack{x\in f^{-1}(1)^r \\ f_\pi^{\oplus r} \textnormal{ is nonconstant} }}\prod_{i=1}^r\max\{1,\Cert(f_{\pi(x)},x^{(i)})\}.
%     $$
% \end{lemma}

We require the following generalization of a result from~\cite{Sav02}. Savick\'y proved that for functions $f_1:\zo^n\to\zo$ and $f_2:\zo^n\to\zo$, it holds that $\dtsize(f_1\oplus f_2)\ge \dtsize(f_1)\cdot\dtsize(f_2)$ \cite[Lemma 2.1]{Sav02}. We will use the following analogous statement for partial functions.

% \calnote{This is Lemma 2.1 in the paper ``On determinism versus unambiquous nondeterminism for decision trees'' by Savick\'y. In the paper, we can add a proper citation. Sieling and HJLT prove a similar lemma, but with the additional claim that one can \textit{efficiently} extract a decision tree from $f$ of size $\sqrt{s}$ from any size-$s$ decision tree for $f^{\oplus 2}$.}
% \lnote{Why don't we just straight up say that this is a generalization of Savicky from total to partial functions? Also, are we sure this partial function statement does not appear explicitly in HJLT or Sieling?}

\begin{theorem}[Generalization of Savick\'y \citep{Sav02}]
    \label{thm:zero-error-xor}
    Let $f^{(1)},\ldots, f^{(r)}$ be functions, $f^{(i)}:D^{(i)}\to\zo$ with $D^{(i)}\sse\zo^{n^{(i)}}$ for each $i=1,\ldots, r$. Then,
    $$
    \dtsize(f^{(1)}\oplus\cdots\oplus f^{(r)})=\prod_{i=1}^r \dtsize(f^{(i)}).
    $$
\end{theorem}

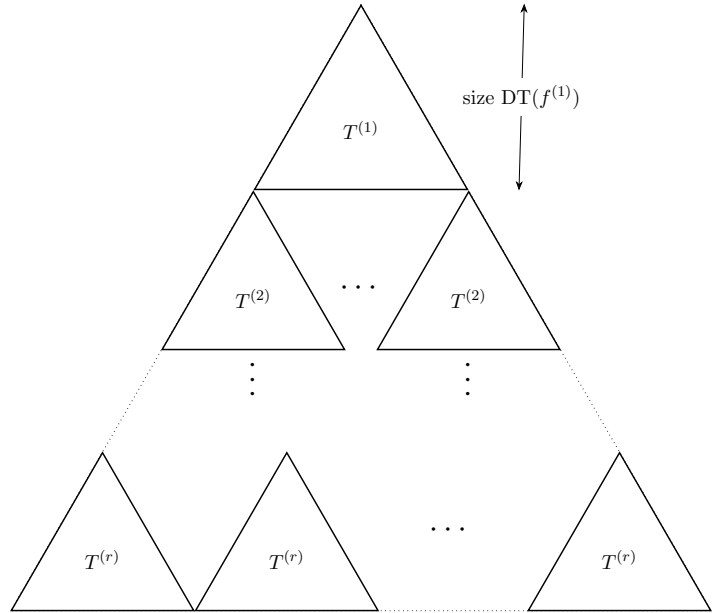
\begin{figure}[!ht]
    \centering
    \scalebox{0.7}{
    \begin{tikzpicture}[tips=proper]
        \node[isosceles triangle,
            draw,thick,
            isosceles triangle apex angle=60,
            rotate=90,
            minimum size=3.5cm] (T1) at (0,0){};
        
        \node[isosceles triangle,
            draw,thick,
            isosceles triangle apex angle=60,
            rotate=90,
            minimum size=3cm,
            anchor=east] (T2L) at (T1.left corner){};

        \node[isosceles triangle,
            draw,thick,
            isosceles triangle apex angle=60,
            rotate=90,
            minimum size=3cm,
            anchor=east] (T2R) at (T1.right corner){};
    
        \node[isosceles triangle,
            draw,
            dotted,
            isosceles triangle apex angle=60,
            rotate=90,
            minimum size=11.52cm,
            anchor=east] (T5) at (T1.east){};
        
        \node[isosceles triangle,
            draw,
            thick,
            isosceles triangle apex angle=60,
            rotate=90,
            minimum size=3cm,
            anchor=east] (T3LL) at ([yshift=-0.00cm]T5.99){};
        \node[isosceles triangle,
            draw,
            thick,
            isosceles triangle apex angle=60,
            rotate=90,
            minimum size=3cm,
            anchor=east] (T3L) at ([xshift=3.5cm,yshift=-0.00cm]T5.99){};
        \node[isosceles triangle,
            draw,thick,
            isosceles triangle apex angle=60,
            rotate=90,
            minimum size=3cm,
            anchor=east] (T3R) at ([yshift=-0.00cm]T5.261){};    
        
        % dot dot dot
        \draw[] (0,-3) node [] {\LARGE $\cdots$};
        \draw[] (1.7,-7.6) node [] {\LARGE $\cdots$};
        \draw[] ([yshift=-0.5cm]T2L.west) node [rotate=90] {\LARGE $\cdots$};
        \draw[] ([yshift=-0.5cm]T2R.west) node [rotate=90] {\LARGE $\cdots$};

        % triangle labels
        \draw[] (T1.center) node [] {$T^{(1)}$};
        \draw[] (T2L.center) node [] {$T^{(2)}$};
        \draw[] (T2R.center) node [] {$T^{(2)}$};
        \draw[] (T3L.center) node [] {$T^{(r)}$};
        \draw[] (T3LL.center) node [] {$T^{(r)}$};
        \draw[] (T3R.center) node [] {$T^{(r)}$};

        % Depth label
        \draw[{Stealth[scale=1]}-{Stealth[scale=1]}] ([xshift=3.1cm]T1.east) to node[midway,fill=white!30,scale=1] {size $\dtsize(f^{(1)})$} ([xshift=3cm]T1.west);
               
    \end{tikzpicture}
    }
  \captionsetup{width=.9\linewidth}

    \caption{Illustration of a stacked decision tree for a function $f^{(1)}\oplus\cdots\oplus f^{(r)}$. For an input $x=(x^{(1)},\ldots,x^{(r)})$, the decision tree sequentially computes $f^{(i)}(x^{(i)})$ for each $i=1,\ldots,r$ using a decision tree $T^{(i)}$ of size $\dtsize(f^{(i)})$ for $f^{(i)}$. Then at the leaf it outputs $f^{(1)}(x^{(1)})\oplus\cdots\oplus f^{(r)}(x^{(r)})$. The overall size of the decision tree is $\prod_{i=1}^r\dtsize(f^{(i)})$.}
    \label{fig:stacked}
\end{figure}

\begin{proof}
    First, the upper bound $\dtsize(f^{(1)}\oplus\cdots\oplus f^{(r)})\le \prod_{i=1}^r \dtsize(f^{(i)})$ follows by considering the decision tree for $f^{(1)}\oplus\cdots\oplus f^{(r)}$ which sequentially computes $f^{(i)}(x)$ for each $i=1,\ldots, r$ using a decision tree of size $\dtsize(f^{(i)})$. See \Cref{fig:stacked} for an illustration of this decision tree.

    The lower bound is by induction on $\sum_{i\in [r]} n^{(i)}$, the total number of input variables. In the base case, $n=0$ and the bound is trivially true: the constant function requires a decision tree of size $1$. For the inductive step, let $T$ be a decision tree for $f^{(1)}\oplus\cdots\oplus f^{(r)}$ of size $\dtsize(f^{(1)}\oplus\cdots\oplus f^{(r)})$, and let $x_{j}$ be the variable queried at the root. Assume without loss of generality that $x_{j}$ belongs to $f^{(1)}$. The subfunctions computed at the left and right branches of the root of $T$ are $f^{(1)}_{x_j\leftarrow 0}\oplus\cdots\oplus f^{(r)}$ and $f^{(1)}_{x_j\leftarrow 1}\oplus\cdots\oplus f^{(r)}$, respectively. Each is a function on $\left(\sum_{i\in [r]} n^{(i)}\right)-1$ many variables and so we can apply the inductive hypothesis. Therefore:
    \begin{align*}
        \dtsize(f^{(1)}\oplus\cdots\oplus f^{(r)})&=|T|\\
        &\ge\dtsize(f^{(1)}_{x_j\leftarrow 0}\oplus\cdots\oplus f^{(r)})+\dtsize(f^{(1)}_{x_j\leftarrow 1}\oplus\cdots\oplus f^{(r)})\tag{Root of $T$ is $x_j^{(1)}$}\\
        &\ge \dtsize(f^{(1)}_{x_j\leftarrow 0})\prod_{i=2}^r \dtsize(f^{(i)})+\dtsize(f^{(1)}_{x_j\leftarrow 1})\prod_{i=2}^r \dtsize(f^{(i)})\tag{Inductive hypothesis}\\
        &=\left(\dtsize(f^{(1)}_{x_j\leftarrow 0})+\dtsize(f^{(1)}_{x_j\leftarrow 1})\right)\prod_{i=2}^r \dtsize(f^{(i)})\\
        &\ge \prod_{i=1}^r \dtsize(f^{(i)})
    \end{align*}
    where the last step follows from the fact that $\dtsize(f^{(1)}_{x_j\leftarrow 0})+\dtsize(f^{(1)}_{x_j\leftarrow 1})\ge \dtsize(f^{(1)})$. Indeed, one can construct a decision tree for $f^{(1)}$ of size $\dtsize(f^{(1)}_{x_j\leftarrow 0})+\dtsize(f^{(1)}_{x_j\leftarrow 1})$ by querying $x_j$ at the root and on the left branch placing a tree for $f^{(1)}_{x_j\leftarrow 0}$ and on the right branch placing a tree for $f^{(1)}_{x_j\leftarrow 1}$.
\end{proof}

\begin{proof}[Proof of \Cref{lem:patchup-xor}]
    Let $\Pi$ denote the set of paths. For each path $\pi\in \Pi$, we write $\pi^{(i)}$ for $i\in [r]$ to denote the part of $\pi$ corresponding to the $i$th block of input variables. This way, the restricted function $f_\pi^{\oplus r}$ corresponds to the function $f_{\pi^{(1)}}\oplus\cdots\oplus f_{\pi^{(r)}}$ Then we have
\begin{align*}
    \dtsize(f_\pi^{\oplus r})&\le \sum_{\pi\in \Pi}\dtsize(f_\pi^{\oplus r})\\
    &= \sum_{\substack{\pi\in\Pi\\ f_\pi^{\oplus r} \text{ is constant}}}\dtsize(f_\pi^{\oplus r})+\sum_{\substack{\pi\in \Pi\\ f_\pi^{\oplus r} \text{ is nonconstant}}}\dtsize(f_\pi^{\oplus r})\\
    &= |T|+\sum_{\substack{\pi\in \Pi\\ f_\pi^{\oplus r} \text{ is nonconstant}}}\dtsize(f_\pi^{\oplus r})\tag{$\dtsize(f_\pi^{\oplus r})=1$ when $f_\pi^{\oplus r}$ is constant}\\
    &= |T|+\sum_{\substack{\pi\in \Pi\\ f_\pi^{\oplus r} \text{ is nonconstant}}}\prod_{i=1}^r\dtsize(f_{\pi^{(i)}})\tag{\Cref{thm:zero-error-xor}}\\
    &\le |T|+\sum_{\substack{\pi\in \Pi\\ f_\pi^{\oplus r} \text{ is nonconstant}}}\prod_{i=1}^r\left(1+\sum_{x: f_{\pi^{(i)}}(x)=1}\Cert(f_{\pi^{(i)}},x)\right)\tag{\Cref{claim:zero-error-patch-up}}
\end{align*}

We use the fact that $1+a\le 2\max\{1,a\}$ for all $a\in\R$ to rewrite the above in a simpler form, while suffering a factor of $2^r$:
\begin{align*}
    |T|+&\sum_{\substack{\pi\in \Pi\\ f_\pi^{\oplus r} \text{ is nonconstant}}}\prod_{i=1}^r\left(1+\sum_{x: f_{\pi^{(i)}}(x)=1}\Cert(f_{\pi^{(i)}},x)\right)\\
    &\le |T|+\sum_{\substack{\pi\in \Pi\\ f_\pi^{\oplus r} \text{ is nonconstant}}}\prod_{i=1}^r\left(2\max\{1,\sum_{x: f_{\pi^{(i)}}(x)=1}\Cert(f_{\pi^{(i)}},x)\}\right)\\
    &\le |T|+2^r\sum_{\substack{\pi\in \Pi\\ f_\pi^{\oplus r} \text{ is nonconstant}}}\prod_{i=1}^r\left(\sum_{x: f_{\pi^{(i)}}(x)=1}\max\{1,\Cert(f_{\pi^{(i)}},x)\}\right)\\
    &= |T|+2^r\sum_{\substack{\pi\in \Pi\\ f_\pi^{\oplus r} \text{ is nonconstant}}}\sum_{x^{(1)}:f_{\pi^{(1)}}(x^{(1)})=1}\cdots\sum_{x^{(r)}:f_{\pi^{(r)}}(x^{(r)})=1}\prod_{i=1}^r\max\{1,\Cert(f_{\pi^{(i)}},x^{(i)})\}\\
    &=|T|+2^r\sum_{\substack{x\in f^{-1}(1)^r \\ f_{\pi(x)}^{\oplus r}\textnormal{ is nonconstant} }}\prod_{i=1}^r\max\{1,\Cert(f_{\pi^{(i)}},x^{(i)})\}
\end{align*}

    % &\le |T|+\sum_{\substack{\pi\in \Pi\\ f_\pi^{\oplus r} \text{ is nonconstant}}}\prod_{i=1}^r\left(2\max\{1,\sum_{x: f_{\pi^{(i)}}(x)=1}\Cert(f_{\pi^{(i)}},x)\}\right)\tag{$1+a\le 2\max\{1,a\}$ for all $a\in\R$}\\
    % &\le |T|+2^r\sum_{\substack{\pi\in \Pi\\ f_\pi^{\oplus r} \text{ is nonconstant}}}\prod_{i=1}^r\left(\sum_{x: f_{\pi^{(i)}}(x)=1}\max\{1,\Cert(f_{\pi^{(i)}},x)\}\right)\\
    % &= |T|+2^r\sum_{\substack{\pi\in \Pi\\ f_\pi^{\oplus r} \text{ is nonconstant}}}\sum_{x^{(1)}:f_{\pi^{(1)}}(x^{(1)})=1}\cdots\sum_{x^{(r)}:f_{\pi^{(r)}}(x^{(r)})=1}\prod_{i=1}^r\max\{1,\Cert(f_{\pi^{(i)}},x^{(i)})\}\\
    % &=|T|+2^r\sum_{\substack{x\in f^{-1}(1)^r \\ f_{\pi(x)}^{\oplus r}\textnormal{ is nonconstant} }}\prod_{i=1}^r\max\{1,\Cert(f_{\pi^{(i)}},x^{(i)})\}
% \end{align*}
    where the last equality follows from the fact that $\Pi$ partitions the input space and so for every $x\in f^{-1}(1)^r$ there is exactly one path $\pi\in \Pi$ such that $f_{\pi^{(i)}}(x^{(i)})=1$ for all $i\in [r]$.
\end{proof}
\section{Proof of \Cref{lem:hard-distribution-lemma-xor}}
\label{sec:hard-distribution-xor}
% \begin{lemma}[Hard distribution lemma for $f^{\oplus r}$]
% \label[lemma]{lem:hard-distribution-lemma-xor}
%     Let $T$ be a decision tree, $f:D\to\zo$ be a nonconstant function and let $C \sse f^{-1}(1)$. There is a distribution $\mathcal{D}$ over the inputs in $C$ and their sensitive neighbors such that for all $r\ge 1$,
%     $$
%     \mathrm{error}_{\mathcal{D}^{\otimes r}}(T,f^{\oplus r})\ge \left(\frac{1}{2|C|\Sens(f)}\right)^r\sum_{\substack{x\in C^r\\ \Sens(f_{\pi(x)}^{\oplus r},x)\neq\varnothing}}\prod_{i=1}^r\max\{1,\Sens(f_{\pi(x)}^{(i)},x^{(i)})\}.
%     $$
% \end{lemma}

\begin{lemma}[Error of $f^{\oplus r}$ conditioned on an input]
\label[lemma]{lem:xor-conditional-error}
    Let $T$ be a decision tree, $r\ge 1$, $f:D\to\zo$ be a nonconstant function, $x\in f^{-1}(1)^r$, $\mathcal{D}$ be the canonical hard distribution, and $\pi$ be the path followed by $x$ in $T$ and assume that $\Sens(f_{\pi(x)}^{\oplus r},x)\neq\varnothing$, then
    $$
    \Prx_{\bz\sim \mathcal{D}^{\otimes r}}[f^{\oplus r}(\bz)\neq T(\bz)\mid \text{ first sampling }x]\ge 2^{-r}\prod_{i=1}^r\frac{\max\{1,|\Sens(f_\pi^{(i)},x^{(i)})|\}}{\max\{1,|\Sens(f^{(i)},x^{(i)})|\}}.
    $$
\end{lemma}

\begin{proof}
    % First, $\Sens(f_\pi^{\oplus r},x)\neq\varnothing$ by our assumption that $f_\pi^{\oplus r}$ is nonconstant. Indeed, since $x\in\minterm(f)^r$, in order to have $\Sens(f_\pi^{\oplus r},x)=\varnothing$, $\pi$ would need to restrict all of the $1$-coordinates of $x$, leaving $0$-coordinates as the only unrestricted coordinates. By monotonicity though this implies $f_\pi^{\oplus r}$ is constant (since $x$ is a $1$-input). 
    Let $y\in \Sens(f_\pi^{\oplus r},x)$ and let $j\in [r]$ be such that $y^{(j)}\in \Sens(f_\pi^{(j)},x^{(j)})$. Since $f^{\oplus r}(y)\neq f^{\oplus r}(x)$, we have
    \begin{align*}
        \Prx_{\bz\sim \mathcal{D}^{\otimes r}}&[f^{\oplus r}(\bz)\neq T(\bz)\mid \text{ first sampling }x]\\
        &\ge \min\big\{ \Prx_{\bz\sim \mathcal{D}^{\otimes r}}[\bz=x\mid \text{ first sampling }x],\\
        &\qquad\Prx_{\bz\sim \mathcal{D}^{\otimes r}}[\bz=(x^{(1)},...,y,...,x^{(r)})\text{ for }y\in\Sens(f_\pi^{(j)},x^{(j)})\mid \text{ first sampling }x] \big\}\tag{Either $x$ or $(x^{(1)},...,y,...,x^{(r)})$ makes an error}\\
        &= \min\left\{2^{-r},2^{-(r-1)}\frac{|\Sens(f_\pi^{(j)},x^{(j)})|}{|\Sens(f^{(j)},x^{(j)})|}\right\}\tag{Definition of $\mathcal{D}^{\otimes r}$}\\
        &\ge 2^{-r}\frac{|\Sens(f_\pi^{(j)},x^{(j)})|}{|\Sens(f^{(j)},x^{(j)})|}\\
        &\ge 2^{-r}\frac{|\Sens(f_\pi^{(j)},x^{(j)})|}{|\Sens(f^{(j)},x^{(j)})|}\cdot\prod_{i\in [r]\setminus\{j\}}\frac{\max\{1,|\Sens(f_\pi^{(i)},x^{(i)})|\}}{\max\{|\Sens(f^{(i)},x^{(i)})|\}}\tag{${\max\{1,|\Sens(f_\pi^{(i)},x^{(i)})|\}}\le \max\{1,|\Sens(f^{(i)},x^{(i)})|\} $ for all $i$}\\
        &=2^{-r}\prod_{i=1}^r\frac{\max\{1,|\Sens(f_\pi^{(i)},x^{(i)})|\}}{\max\{1,|\Sens(f^{(i)},x^{(i)})|\}}\tag{$|\Sens(f_\pi^{(j)},x^{(j)})|\neq 0$ by assumption}.
    \end{align*}
\end{proof}

We can now prove the main result of this section.

\begin{proof}[Proof of \Cref{lem:hard-distribution-lemma-xor}]
    As in the case of \Cref{lem:hard-distribution-lemma}, we assume without loss of generality that $C=f^{-1}(1)$. We lower bound the error as follows 
    \begin{align*}
        \eps&\ge \Prx_{\bz\sim\mathcal{D}^{\otimes r}}[T(\bz)\neq f^{\oplus r}(\bz)]\\
        &=\sum_{\substack{x\in C^r \\ \Sens(f_{\pi(x)}^{\oplus r},x)\neq\varnothing }}\frac{1}{|C|^r}\cdot\Prx_{\bz\sim\mathcal{D}^{\otimes r}}[T(\bz)\neq f^{\oplus r}(\bz)\mid \text{ first sampling }x\text{ from } C^r]\\
        &\ge (2|C|)^{-r}\sum_{\substack{x\in C^r \\ \Sens(f_{\pi(x)}^{\oplus r},x)\neq\varnothing }}\prod_{i=1}^r\frac{\max\{1,|\Sens(f_\pi^{(i)},x^{(i)})|\}}{\max\{1,|\Sens(f^{(i)},x^{(i)})|\}}\tag{\Cref{lem:xor-conditional-error}}\\
        &\ge (2|C|\cdot\Sens(f))^{-r}\sum_{\substack{x\in C^r \\ \Sens(f_{\pi(x)}^{\oplus r},x)\neq\varnothing }}\prod_{i=1}^r\max\{1,|\Sens(f_\pi^{(i)},x^{(i)})|\}\tag{$|\Sens(f^{(i)},x^{(i)})|\le\Sens(f)$}
    \end{align*}
    which completes the proof.
\end{proof}

\section{Proof of \Cref{thm:dt-learn-xor}}
\label{sec:dt-learn-xor}

% The main theorem of this section is the following reduction from approximating {\sc Vertex Cover} to {\sc DT-Learn}. In this reduction, we are able to achieve better hardness of approximation at the cost of increasing the runtime of the reduction. This reduction enables us to prove \Cref{thm:main,thm:main-general}.

% \begin{theorem}[Main reduction from vertex cover to {\sc DT-Learn}]
% \label{thm:dt-learn-xor}
%     For all $r\ge 1$, $\eps<2^{-3r}$, and $A>1$ the following holds. If there is a time $t(s,1/\eps)$ algorithm for solving {\sc DT-Learn}$(s,A\cdot s,\eps)$ on $n$-variable functions with $s=O(n^r)$, then {\sc Vertex Cover} can be $(1+\delta')$-approximated on degree-$d$, $n$-vertex graphs in randomized time $O(rn^2\cdot t(n^{2r},1/\eps))$ for any $\delta'>(A^{1/r}-1+\eps^{1/r}\cdot 8)d+A^{1/r}-1$.
% \end{theorem}

Before proving this theorem, we establish a few properties of $\ell\text{-}\IsEdge^{\oplus r}$ which will be helpful for our analysis.

\begin{theorem}[Decision tree size lower bound for computing $\ell\text{-}\IsEdge^{\oplus r}$]
\label{thm:ellisedge-lb-xor}
Let $T$ be a decision tree for $\ell\text{-}\IsEdge_G^{\oplus r}$ with $\ell, r\ge 1$. Let $k$ be the minimum vertex cover size of $G$ and let $m$ denote the number of edges of $G$. Then, if $\error_{\ell\text{-}\mathcal{D}_G^{\otimes r}}(T,\ell\text{-}\IsEdge^{\oplus r})\le \eps$ for the canonical hard distribution $\ell\text{-}\mathcal{D}_G$, we have
    $$
    |T|\ge \big[(\ell+1)(k+m)\big]^r-\eps \big[8m(\ell+1)\big]^r
    $$
\end{theorem}

\begin{proof}
    Since $\ell\text{-}\IsEdge$ has $m$ many $1$-inputs over the dataset $\ell\text{-}D_G$ and $\Sens(\ell\text{-}\IsEdge)=2(\ell+1)$, we have
    \begin{align*}
        \eps&\ge \left(\frac{1}{4m(\ell+1)}\right)^r\sum_{\substack{x\in \ell\text{-}D_G^r\\ \Sens(\ell\text{-}\IsEdge_{\pi(x)}^{\oplus r},x)\neq\varnothing}}\prod_{i=1}^r\max\{1,\Sens(\ell\text{-}\IsEdge_{\pi(x)}^{(i)},x^{(i)})\}\tag{\Cref{lem:hard-distribution-lemma-xor}}\\
        &=\left(\frac{1}{4m(\ell+1)}\right)^r\sum_{\substack{x\in \ell\text{-}D_G^r\\ \Sens(\ell\text{-}\IsEdge_{\pi(x)}^{\oplus r},x)\neq\varnothing}}\prod_{i=1}^r\max\{1,\Cert(\ell\text{-}\IsEdge_{\pi(x)}^{(i)},x^{(i)})\}\tag{\Cref{prop:sens-equals-cert}}\\
        &\ge \left(\frac{1}{4m(\ell+1)}\right)^r\sum_{\substack{x\in \ell\text{-}D_G^r\\ \ell\text{-}\IsEdge_{\pi(x)}^{\oplus r} \text{ is nonconstant}    }}\prod_{i=1}^r\max\{1,\Cert(\ell\text{-}\IsEdge_{\pi(x)}^{(i)},x^{(i)})\}\\
        &\ge \left(\frac{1}{8m(\ell+1)}\right)^r(\dtsize(\ell\text{-}\IsEdge^{\oplus r})-|T|)\tag{\Cref{lem:patchup-xor}}.
    \end{align*}
    In this derivation, we used the fact that if an input $x\in \ell\text{-}D_G^r$ is such that $\ell\text{-}\IsEdge_{\pi(x)}^{\oplus r}$ is nonconstant, then it must be the case that there is some block $i\in [r]$ where the path $\pi$ does not fully restrict the sensitive coordinates in the edge indicator for block $i$, and therefore it must also be the case that $\Sens(\ell\text{-}\IsEdge_{\pi(x)}^{\oplus r},x)\neq\varnothing$. Now, we can rearrange this lower bound on $\eps$ to obtain:
    \begin{align*}
        |T|&\ge\dtsize(\ell\text{-}\IsEdge^{\oplus r})-\eps \big[8m(\ell+1)\big]^r\\
        &=\dtsize(\ell\text{-}\IsEdge)^r-\eps \big[8m(\ell+1)\big]^r\tag{\Cref{thm:zero-error-xor}}\\
        &\ge \big[(\ell+1)(k+m)\big]^r-\eps \big[8m(\ell+1)\big]^r\tag{\Cref{lem:ell-isedge-lb-zero-error}}
    \end{align*}
    which completes the proof.
\end{proof}

The following proposition allows us to translate the above lower bound into a slightly simpler form.
\begin{proposition}
\label{prop:inequality-for-xor}
    For all $a,b,r>0$ such that $a\ge b$, we have $a^r-b^r\ge (a-b)^r$.
\end{proposition}

\begin{proof}
    Since $a\ge b$, we have
    $$
    1\ge \left(1-\frac{b}{a}\right)^r+\left(\frac{b}{a}\right)^r.
    $$
    Multiplying both sides of the inequality by $a^r$ and rearranging gives the desired bound.
\end{proof}

With \Cref{thm:ellisedge-lb-xor,prop:inequality-for-xor}, we are able to prove the main technical lemma used for our reduction.

\begin{lemma}[Main technical lemma for \Cref{thm:dt-learn-xor}]
\label{lem:technical-xor}
For all $\delta,\delta',\eps >0$ and $d,k,r\ge 1$, the following holds. Given a constant degree-$d$ graph $G$ with $m$ edges, $n$ vertices, and parameter $k$, there is a choice of $\ell=\Theta(n)$ and a polynomial-time computable quantity $s\in\N$ such that so long as $\delta'>(\delta+8\eps^{1/r})d+\delta$, $dk\ge m$, and $\eps<2^{-3r}$ we have:
\begin{itemize}
    \item \textbf{Yes case}: if $G$ has a vertex cover of size at most $k$, then there is a decision tree of size at most $s$ which computes $\ell\text{-}\IsEdge^{\oplus r}:\zo^{r(n\ell+n)}\to\zo$; and
    \item \textbf{No case}: if every vertex cover of $G$ has size at least $(1+\delta')k$, then $(1+\delta)^r s<|T|$ for any decision tree $T$ with $\error_{\ell\text{-}\mathcal{D}_G^{\otimes r}}(T,\ell\text{-}\IsEdge)\le \eps$.
\end{itemize}
\end{lemma}

\begin{proof}
    Given a degree-$d$, $m$-edge, $n$-vertex graph $G$ and parameter $k$, we choose $\ell=\Theta(n)$ so that $\delta'>(\delta+8\eps^{1/r})d+\delta+\frac{(1+\delta)mn}{k(\ell+1)}$ and set $s=\left[(\ell+1)(k+m)+mn\right]^r$. Note that such an $\ell$ exists since $k=\Theta(n)$ for constant-degree graphs. We now prove the two points separately.

    \paragraph{Yes case.}{
        In this case, we have
        \begin{align*}
            \dtsize(\ell\text{-}\IsEdge^{\oplus r})&=\dtsize(\ell\text{-}\IsEdge)^r\tag{\Cref{thm:zero-error-xor}}\\
        &\le \left[(\ell+1)(k+m)+mn\right]^r=s\tag{\Cref{thm:ell-isedge-ub}}.
        \end{align*}
    }
    \paragraph{No case.}{
        In this case, let $T$ be a decision tree with $\error_{\ell\text{-}\mathcal{D}_G^{\otimes r}}(T,\ell\text{-}\IsEdge)\le \eps$. Then we have
        \begin{align*}
            (1+\delta)^r s&=\bigg[(1+\delta)[(\ell+1)(k+m)+mn]\bigg]^r\\
            &<\left[(\ell+1)(1+\delta')k+(\ell+1)(1-8\eps^{1/r}) m\right]^r\tag{\Cref{prop:inequality} with $\alpha=8\eps^{1/r}$}\\
            &\le \big[(\ell+1)(k+m)\big]^r-\eps \big[8m(\ell+1)\big]^r\tag{\Cref{prop:inequality-for-xor}}\\
            &\le |T|\tag{\Cref{thm:ellisedge-lb-xor}}.
        \end{align*}
    }
    We've shown the desired bounds in both the Yes and No cases so the proof is complete.
\end{proof}

\begin{proof}[Proof of \Cref{thm:dt-learn-xor}]
Let $\mathcal{A}$ be the algorithm for {\sc DT-Learn} from the theorem statement. Given an $n$-vertex, $m$-edge graph $G$ of constant degree $d$, we'll use $\mathcal{A}$ to approximate {\sc Vertex Cover} on $G$.

\paragraph{The reduction.}{
    First, we check whether $dk\ge m$. If $dk<m$, our algorithm outputs ``No'' as $G$ cannot have a vertex cover of size at most $k$ (see \Cref{fact:constant degree graphs have large vc size}). Otherwise, we proceed under the assumption that $dk\ge m$. Let $s\in \N$ be the quantity from \Cref{lem:technical-xor}. We run $\mathcal{A}$ over the distribution $\ell\text{-}\mathcal{D}_G^{\otimes r}$ and on the function $\ell\text{-}\IsEdge^{\oplus r}:\zo^N\to\zo$ where $\ell$ is as in \Cref{lem:technical-xor}. Note that $N=r\cdot(n\ell+n)=O(rn^2)$ and $s=O(n^{2r})$. See \Cref{fig:solving vc with dt learn xor} for the exact procedure we run.
        % Given an $n$-vertex, $m$-edge graph $G$ of constant degree $d$, and a parameter $k$, we construct the following instance of {\sc DT-Learn}. Fix $\ell=\Theta(n)$ as in \Cref{lem:technical-xor}. Consider the function $\ell\text{-}\IsEdge^{\oplus r}_{G}:\zo^{rN}\to\zo$ where $N=n+\ell n$. Let $\mathcal{D}$ be the canonical hard distribution for $\ell\text{-}\IsEdge_G$. We run the algorithm for {\sc DT-Learn} on $\ell\text{-}\IsEdge^{\oplus r}_{G}$ over the distribution $\mathcal{D}^{\otimes r}$ with $s=\left[(\ell+1)(k+m)+mn\right]^r\le O(N^r)$ for $t(s)\le t(N^r)$ time steps with error parameter $\eps$ from the theorem statement. Let $T$ denote the decision tree output by the learner. Our algorithm for vertex cover outputs ``Yes'' if and only if $|T|\le A\cdot\left[(\ell+1)(k+m)+mn\right]^r$ and the error of $T$ with $\ell\text{-}\IsEdge^{\oplus r}_{G}$ over $\mathcal{D}^{\otimes r}$ is less than $\eps$. 
}

\begin{figure}[!ht]
\begin{tcolorbox}[colback = white,arc=1mm, boxrule=0.25mm]
\vspace{3pt}
$\textsc{Vertex Cover}(k,(1+\delta')\cdot k)$:
\begin{itemize}[leftmargin=20pt,align=left]
\item[\textbf{Given:}] $G$, an $m$-edge degree-$d$ graph over $n$ vertices and $k\in \N$
\item[\textbf{Run:}]\textsc{DT-Learn}$(s,A\cdot s,\eps)$ for $t(s,1/\eps)$ time steps providing the learner with
\begin{itemize}[align=left,labelsep*=0pt]
    \item \textit{queries}: return $\ell\text{-}{\mathrm{\IsEdge}}^{\oplus r}(x)$ for a query $x\in\zo^N$; and
    \item \textit{random samples}: return $\bx\sim\ell\text{-}\mathcal{D}_G^{\otimes r}$ for a random sample.
\end{itemize}
\item[$T_{\text{hyp}}\leftarrow$ decision tree output of the learner]
\item[$\eps_{\text{hyp}}\leftarrow \mathrm{error}_{\ell\text{-}\mathcal{D}_G^{\otimes r}}(T_{\text{hyp}},\ell\text{-}{\mathrm{\IsEdge}}^{\oplus r})$]
\item[\textbf{Output:}] \textsc{Yes} if and only if $|T_{\text{hyp}}|\le A\cdot s$ and $\eps_{\text{hyp}}\le \eps$
\end{itemize}
\vspace{3pt}
\end{tcolorbox}
\medskip
\caption{Using an algorithm for \textsc{DT-Learn} on $\ell\text{-}{\mathrm{\IsEdge}}^{\oplus r}$ to solve \textsc{Vertex Cover}.}
\label{fig:solving vc with dt learn xor}
\end{figure}

    \paragraph{Runtime.}{Any query to $\ell\text{-}\IsEdge^{\oplus r}_{G}$ can be answered in $O(N)$ time. Similarly, a random sample can be obtained in $O(N)$ time. The algorithm requires time $O(N\cdot t(N^r,1/\eps))$ to run the learner and then $O(N^2)$ time to compute $\dist_{\mathcal{D}^{\otimes r}}(T,\ell\text{-}\IsEdge^{\oplus r}_{G})$. This implies an overall runtime of $O(N\cdot t(N^r,1/\eps))=O(rn^2\cdot t(n^{2r},1/\eps))$.}

    \paragraph{Correctness.}{
        If we let $\delta\coloneqq A^{1/r}-1$, then the assumption of the theorem statement is that $\delta'>(\delta+8\eps^{1/r})d+\delta$. Therefore, we are able to apply \Cref{lem:technical-xor} from which we deduce correctness.
        
        In the \textbf{Yes case}, if $G$ has a vertex cover of size at most $k$, then there is a decision tree of size at most $s$ computing $\ell\text{-}\IsEdge^{\oplus r}_{G}$. So by the guarantees of {\sc DT-Learn}, our algorithm correctly outputs ``Yes''. }

        In the \textbf{No case}, every vertex cover of $G$ has size at least $(1+\delta')k$. If $\eps_{\text{hyp}}>\eps$ then our algorithm for {\sc Vertex Cover} correctly outputs ``No''. Otherwise, assume that $\eps_{\text{hyp}}\le \eps$. Then, \Cref{lem:technical-xor} ensures that $(1+\delta)^rs<|T_{\text{hyp}}|$ and so our algorithm correctly outputs ``No'' in this case as well.
    % In the \textbf{yes} case, there is a vertex cover of size $\le k$ for $G$ and so 
    % \begin{align*}
    %     \dtsize(\ell\text{-}\IsEdge^{\oplus r}_{G_{\mathrm{yes}}})&= \dtsize(\ell\text{-}\IsEdge_{G_{\mathrm{yes}}})^r\tag{\Cref{thm:zero-error-xor}}\\
    %     &\le \left[(\ell+1)(k+m)+mn\right]^r\tag{\Cref{thm:ell-isedge-ub}}.
    % \end{align*}
    % Therefore, by the approximation guarantee of our learning algorithm, we have $|T|\le A\cdot\left[(\ell+1)(k+m)+mn\right]^r$ and the error is less than $\eps$. So our algorithm for vertex cover correctly outputs ``Yes''.

    % In the \textbf{no} case, every vertex cover of $G$ has size at least $(1+\delta')k$. If $\dist_{\mathcal{D}^{\otimes r}}(T,\ell\text{-}\IsEdge^{\oplus r}_{G})\ge \eps$, then our algorithm correctly outputs ``No''. Otherwise, assume that $\dist_{\mathcal{D}^{\otimes r}}(T,\ell\text{-}\IsEdge^{\oplus r}_{G})<\eps$. We'll show that $|T|>A\cdot\left[(\ell+1)(k+m)+mn\right]^r$ so that our algorithm correctly outputs ``No'' in this case as well. If we let $\delta\coloneqq A^{1/r}-1$, then the assumption of the theorem statement is that $\delta'>(\delta+8\eps^{1/r})d+\delta$. Thus, we can apply \Cref{lem:technical-xor} to get that
    % $$
    % A\cdot\dtsize(\ell\text{-}\IsEdge^{\oplus r}_{G_{\mathrm{yes}}})=(1+\delta)^r\dtsize(\ell\text{-}\IsEdge^{\oplus r}_{G_{\mathrm{yes}}})<\dtsize_{\mathcal{D}^{\otimes r}}(\ell\text{-}\IsEdge^{\oplus r}_{G_{\mathrm{no}}},\eps).
    % $$
    % Since $\left[(\ell+1)(k+m)+mn\right]^r=\dtsize(\ell\text{-}\IsEdge^{\oplus r}_{G_{\mathrm{yes}}})$ this ensures that our algorithm correctly outputs ``No'' in this case. 
\end{proof}

\begin{remark}[Why we require such sharp lower bounds in the proof of \Cref{thm:dt-learn-xor}]
\label[remark]{remark:sharp}
    A key step in the analysis of the correctness of our reduction is \Cref{lem:technical-xor}. Since our upper bound for $\ell$-$\IsEdge$ is of the form $s^r$, we require an equally strong lower bound of the form $(s')^r$. A weaker lower bound $(s')^{cr}$ for some $c<1$ would be insufficient, since the parameter $s$ would no longer separate the Yes and No cases in \Cref{lem:technical-xor}. 
\end{remark}

\section{Decision tree minimization given a subset of inputs}
\label{sec:dt-dataset-min}

In this section, we show how our results yield new lower bounds for minimizing decision trees. First, we recall the problem of decision tree minimization \citep{ZB00,Sie08}.

\begin{definition}[Decision tree minimization]
\label{def:dtmin}
    {\sc DT-Min}$(s,s')$ is the following. Given a decision tree $T$ over $n$ variables and parameters $s,s'\in \N$, distinguish between
    \begin{itemize}
        \item \textbf{Yes case}: there is a size-$s$ decision tree $T'$ such that $T'(x)=T(x)$ for all $x\in \zo^n$; and
        \item \textbf{No case}: all decision trees $T'$ such that $T'(x)=T(x)$ for all $x\in \zo^n$ have size at least $s'$.
    \end{itemize}
\end{definition}

\cite{Sie08} proves the following hardness results for {\sc DT-Min}:

\begin{theorem}[Hardness of approximating {\sc DT-Min} \citep{Sie08}]
\label{thm:hardness-dtmin}
    The following hardness results hold for {\sc DT-Min}:
    \begin{itemize}
        \item for all constants $C>1$, {\sc DT-Min}$(s,Cs)$ is \textnormal{NP}-hard; and
        \item for all constants $\gamma<1$, there is no quasipolynomial time algorithm for {\sc DT-Min}$(s,2^{(\log s)^\gamma}\cdot s)$ unless $\mathrm{NP} \sse\mathrm{DTIME}(n^{\polylog (n)})$
    \end{itemize}
\end{theorem}

We observe that our proof of \Cref{thm:main-general} recovers \Cref{thm:hardness-dtmin} and also strengthens the hardness results to hold even when the no case in \Cref{def:dtmin} is strengthened to: there is an explicit set of inputs $D$ and an explicit distribution $\mathcal{D}$ over $D$ such that any decision tree $T'$ which agrees with $T$ with probability $1-\eps$ for $\bx\sim\mathcal{D}$ has size at least $s'$. This is a strict strengthening since any decision tree $T'$ such that $T'(x)=T(x)$ for all $x\in \zo^n$ also agrees with $T$ over the distribution~$\mathcal{D}$.

\begin{theorem}[Hardness of approximating {\sc DT-Dataset-Min}]
\label{thm:hardness-dataset-min}
    Let {\sc DT-Dataset-Min}$(s,s')$ be the variant of {\sc DT-Min}$(s,s')$ where the input includes a subset of inputs $D\sse \zo^n$, the pmf of a distribution $\mathcal{D}$ over $D$, and a parameter $\eps$, and the No case is changed to ``all decision trees $T'$ such that $T'(\bx)=T(\bx)$ with probability $1-\eps$ for $\bx\sim\mathcal{D}$ have size at least $s'$.'' Then the following hardness results hold
    \begin{itemize}
        \item for all constants $C>1$ there is a constant $\eps>0$ such that {\sc DT-Dataset-Min}$(s,Cs)$ with error parameter $\eps$ is \textnormal{NP}-hard; and
        \item for all constants $\gamma<1$, there is a parameter $\eps=2^{-(\log s)^\gamma}$ such that there is no quasipolynomial time algorithm for {\sc DT-Min}$(s,2^{(\log s)^\gamma}\cdot s)$ with error parameter $\eps$ unless $\mathrm{NP} \sse\mathrm{DTIME}(n^{\polylog (n)})$
    \end{itemize}
\end{theorem}

\begin{proof}
    These hardness results follow from the reduction in \Cref{thm:dt-learn-xor}. Specifically, we construct a decision tree $T^\star$ computing $\ell\text{-}\IsEdge^{\oplus r}$ over $\zo^{r(n\ell+n)}$. The set of all $n$ vertices of $G$ trivially forms a vertex cover of $G$. Therefore, we can apply \Cref{thm:ell-isedge-ub} to obtain a decision tree for $\ell\text{-}\IsEdge$ of size $(\ell+1)(n+m)+mn$. We can stack $r$ independent copies of this decision tree as in the proof of \Cref{thm:zero-error-xor} (see \Cref{fig:stacked}) to get a decision tree for $\ell\text{-}\IsEdge^{\oplus r}$ whose size is $\left[(\ell+1)(n+m)+mn\right]^r$. We then choose $\ell\text{-}D_G^r=\supp(\ell\text{-}\mathcal{D}_G^{\otimes r})$ to be the subset of inputs for the minimization instance. Moreover, it is straightforward to compute the pmf of the distribution $\ell\text{-}\mathcal{D}_G^{\otimes r}$ and provide this to the algorithm for {\sc DT-Dataset-Min}. 
    
    Therefore, as in the proof of \Cref{thm:main}, $(1+\delta')$-approximating {\sc Vertex Cover} reduces in polynomial-time to {\sc DT-Dataset-Min}$(s,Cs)$. This completes the proof of the first point in the theorem statement. 

    For the second point, let $\gamma<1$ be given. We choose $r$ large enough so that $(1+\delta)^r>2^{(\log s)^\gamma}$ where $s$ and $\delta$ are parameters from \Cref{lem:technical-xor}. Since $s=O(n^{2r})$, any $r=\polylog n$ satisfying $r^{1-\gamma}\ge \Omega((\log n)^{\gamma})$ is sufficient. For this choice of $r$, our reduction runs in quasipolynomial-time and reduces $(1+\delta')$-approximating {\sc Vertex Cover} to {\sc DT-Dataset-Min}$(s,2^{(\log s)^\gamma}\cdot s)$. Therefore, the proof is complete.
\end{proof}

\end{document}